\documentclass{LMCS}

\def\dOi{10(2:16)2014}
\lmcsheading%
{\dOi}
{1--22}
{}
{}
{Oct.~31, 2013}
{Jun.~27, 2014}
{}


\ACMCCS{[{\bf Theory of computation}]: Computational complexity and cryptography---Proof complexity}

\usepackage{hyperref}

\usepackage{latexsym}
\usepackage{amssymb}
\usepackage{pstricks}
\usepackage{bussproofs}

\theoremstyle{plain}

\newcommand{\rddots}{{\mathinner{\mkern1mu\raise1pt\vbox{\kern7pt\hbox{.}}%
        \mkern2mu\raise4pt\hbox{.}\mkern2mu\raise7pt\hbox{.}\mkern1mu}}}
\newcommand{\proofdots}{{\ddots\vdots\,\rddots}}

\def\Stone{{\hbox{\it Stone}}}
\def\negp{{\overline p}}
\def\negr{{\overline r}}
\def\negx{{\overline x}}
\def\Tpool{{\mathrm{pool}}}

\def\Tdom{{\mathrm{dom}}}
\def\Tmaxdom{{\mathrm{max\,dom}}}

\def\pprime{{\prime\prime}}

\title{Small Stone in Pool}

\author[S.~Buss]{Samuel R. Buss\rsuper a}	%required
\address{{\lsuper a}Department of Mathematics \\
University of California, San Diego\\
La Jolla, CA 92093-0112, USA}	%required
\email{sbuss@math.ucsd.edu}  %optional
\thanks{Supported in part by NSF grants DMS-1101228 and CCR-1213151.}

\author[L.~Ko{\l}odziejczyk]{Leszek Aleksander Ko{\l}odziejczyk\rsuper
b}
\address{{\lsuper b}Institute of Mathematics \\
University of Warsaw \\
Banacha 2, 02-097 Warszawa, Poland}	%optional
\email{lak@mimuw.edu.pl}  %optional
\thanks{{\lsuper{a,b}}This
work was carried out while
the second author was visiting
the University of California, San Diego,
supported by Polish Ministry of Science and
Higher Education programme ``Mobilno\'s\'c Plus''
with additional support from
a grant from the Simons Foundation (\#208717 to Sam Buss).}	%optional

\begin{document}

\begin{abstract}
\noindent The Stone tautologies are known to have polynomial size
resolution refutations and require exponential size regular refutations.
We prove that the Stone tautologies
also have polynomial size proofs in both pool
resolution and the proof system
of regular tree-like resolution with
input lemmas (regRTI).  Therefore, the Stone tautologies
do not separate resolution from DPLL with clause learning.
\end{abstract}

\maketitle

\hfill{\it I have said it thrice; What I tell you three times is true.} \\
\hbox{\relax}\hfill {\sl Fit the First --- The Landing; The Hunting of the Snark} \\
\hbox{\relax}\hfill Lewis Carroll \\

\section{Introduction}\label{sec:intro}

The Davis-Putnam-Logemann-Loveland (DPLL) proof search
method~\cite{DLL:theoremproving,DavisPutnam:procedure},
augmented with clause learning~\cite{SilvaSakallah:GRASP},
has become a core method for solving the satisfiability
(SAT) problem, especially for large-scale instances of
SAT that arise in industrial applications.  However, when
restarts are not allowed, the
proof strength of DPLL with clause learning relative to full resolution remains
unknown.  On one hand, if $\Gamma$ is a set of
clauses, and DPLL with clause learning can show
that $\Gamma$ is unsatisfiable in
$n$ steps, then $\Gamma$ has a resolution refutation
with size polynomially bounded by~$n$ (see \cite{BKS:clauselearning}).
On the other hand, the results
of \cite{AJPU:regularresolution,%
Urquhart:regularresolution,%
BonetBuss:poolVsRegularSAT,%
BBJ:poolVsRegular2}
imply that the length of DPLL with clause learning
proof searches can be nearly exponentially smaller
than the length of the shortest regular resolution
proofs. Systems designed to correspond to
DPLL with clause learning,
such as pool resolution (\cite{VanGelder:PoolResolution})
and regRTI  (\cite{BHJ:ResTreeLearning}),
are therefore simulated by resolution and strictly stronger
than regular resolution.
Determining the exact strength of those systems
is an open problem.

The two papers \cite{AJPU:regularresolution,%
Urquhart:regularresolution} gave examples of three
principles which have polynomial size resolution
refutations, but require exponential size
regular refutations.  In the terminology
of~\cite{BonetBuss:poolVsRegularSAT,%
BBJ:poolVsRegular2}, these three principles were
(1)~the guarded graph tautologies,
(2)~the Stone tautologies, and
(3)~the guarded pebbling tautologies.
Subsequently, \cite{BonetBuss:poolVsRegularSAT,%
BBJ:poolVsRegular2} showed that the guarded graph
tautologies and the guarded pebbling tautologies
have polynomial size pool and regRTI refutations,
and hence can be refuted by polynomially long
DPLL search with clause learning and without
restarts.

It remained open whether the same holds for the
Stone tautologies.  There seems to be
an inherent simplicity in the irregularity
introduced by ``guarded'' versions of
combinatorial principles, such as~(1) and~(3).
This is because
the guarded principles have refutations in which all
irregularities are
at the initial inferences; namely,
the resolution refutation can start by using
resolution to remove the
guard literals, and then give
a (regular) refutation of the
underlying principle as usual.  In contrast, the
prior known resolution refutations for the Stone principles
of~\cite{AJPU:regularresolution}
use irregularity in a more essential fashion, with
the irregularities distributed throughout the
refutation.  Because of this, it was conjectured that
the Stone principles might be examples where
the pool and regRTI systems, and thus DPLL with clause learning and
no restarts,
require exponential size refutations.  That is,
the Stone tautologies were viewed as
candidates for separating DPLL with
clause learning from resolution.

The present paper, however, refutes this conjecture
and establishes that the Stone tautologies do indeed have polynomial size
pool refutations and regRTI refutations. In light of this,
the possibility that pool and regRTI actually simulate the full
power of resolution perhaps becomes slightly more plausible.
Nevertheless, even if such a general simulation result does hold,
it is far from clear whether the methods  we use to deal with the Stone tautologies
can be of much help in proving it.

The remainder of this introduction gives a review
of the basic definitions,
first of the Stone principles, then of the various proof systems.
It concludes by stating our main theorems about the existence
of pool and regRTI refutations.  The reader is encouraged to
consult the introductory sections
of~\cite{BonetBuss:poolVsRegularSAT,BBJ:poolVsRegular2} for
a more extensive discussion of
prior work, and to consult \cite{AJPU:regularresolution} for more on the
Stone principles.  A good general introduction to
DPLL with clause learning is~\cite{BKS:clauselearning}.

\begin{defi}
A {\em literal} is either a propositional variable~$x$
or a negated variable~$\negx$.  A {\em clause} is a set
of literals, usually written as a list of literals separated
by either $\lor$'s (disjunctions) or commas.
A clause is interpreted as the disjunction of its members.
A set~$\Gamma$ of clauses is interpreted as the conjunction
of its members, so $\Gamma$~represents a propositional
formula in conjunctive normal form.
\end{defi}

The next definition describes the Stone
principle of~\cite{AJPU:regularresolution} as a set of clauses.
The Stone principle is a kind of induction principle.
For a given directed acyclic graph (dag), it
states that if each source vertex is pebbled with a red stone
and if each vertex
whose immediate predecessors are pebbled with red stones
is also pebbled with a red stone, then the sink vertex is
pebbled with a red stone.

\begin{defi}
Assume that $G=(V,E)$ is a dag with
a single sink, with vertices $V=\{1,\ldots,N\}$, such that
each non-source vertex of~$G$ has in-degree~2. We assume that
vertices are numbered consistently with the directions
of the edges of~$G$ so that if
there is an edge $(i',i)\in E$ from $i'$ to~$i$ then $i'>i$,
and so that the source nodes
of~$G$ are exactly vertices $n{+}1,n{+}2,\ldots,N$ for some~$n$.
Vertex~$1$ is the sole
sink of~$G$.
Further assume that $m\ge N$;
here, $m$~is the number of ``stones''.
The
(negation of the) Stone tautology for $G$ and~$m$
is denoted $\Stone(G,m)$ and uses variables~$p_{i,j}$ to
indicate that vertex $i\in V$ is marked (``pebbled'') with
the $j$-th stone and variables~$r_j$ to indicate that
the $j$-th stone is colored red.
$\Stone(G,m)$
contains the following clauses:
\begin{itemize}
\item $\bigvee_{j=1}^m p_{i,j}$, for each vertex $i$ in $G$.
(Each vertex is pebbled by at least one
stone.)
\item $\negp_{i,j} \lor r_j$, for each $j = 1,\ldots, m$,
and each source vertex $i$ in $G$.
(Each stone on a source vertex
is colored red.)
\item $\negp_{1,j}\lor \negr_j$, for each $j = 1,\ldots, m$.
(The sink vertex~$1$ is not pebbled by
any red stone.)
\item
$\negp_{i',j'}\lor \negr_{j'}\lor \negp_{i'',j''}\lor \negr_{j''}\lor \negp_{i,j}\lor r_j$,
whenever $i'$ and $i''$ are the two vertices
such that $(i',i) \in E$ and $(i'',i) \in E$ and $j\notin\{j',j''\}$.
(If the two predecessors of~$i$ in~$G$
are pebbled by red stones, then every stone pebbling vertex~$i$ is
also red.
These ``induction clauses'' are equivalent to
$p_{i',j'} \wedge r_{j'} \wedge p_{i'',j''} \wedge r_{j''} \wedge p_{i,j} \rightarrow r_j$.)
\end{itemize}
\end{defi}
\noindent It is permitted that vertices are pebbled with more than one
stone; likewise, the same stone may pebble multiple vertices.

The Stone clauses are clearly inconsistent since if the source vertices
are pebbled with red stones then the induction clauses imply
that all other vertices are also pebbled with  red stones,
and this contradicts the third group of clauses asserting
that the sink vertex is not pebbled with a red stone.

\medskip

\noindent
We next recall the definitions of various types of resolution.

\begin{defi}
Let $A$, $B$, and $C$ be clauses, and $x$ a
literal such that $\negx \notin A$ and $x \notin B$. Consider the inference
\begin{prooftree}
\AxiomC{A} \AxiomC{B}
\BinaryInfC{C}
\end{prooftree}
The literal~$x$ is the {\em resolution variable}.  Three kinds
of inferences are defined by:
\begin{description}
\item[Resolution rule]  We have $x \in A$, $\negx \in B$, and
$C=(A\setminus\{x\})\lor(B\setminus\{\overline{x}\})$.
\item[Degenerate resolution rule] \cite{HBPvG:clauselearn,VanGelder:PoolResolution}
If $x\in A$ and $\overline{x}\in B$,
then $C$ is obtained as in the resolution rule.
If $x\in A$ and $\negx\notin B$,
then $C$ is~$B$.
If $x\notin A$ and $\negx \in B$,
then $C$ is~$A$.
Otherwise $C$ is one of $A$ or~$B$.
\item[w-resolution rule] \cite{BHJ:ResTreeLearning}
The clause $C$ equals
$(A\setminus\{x\})\lor(B\setminus\{\overline{x}\})$.
\end{description}
\end{defi}
The three different types of resolution
coincide when $x\in A$ and $\negx\in B$, in which case
we refer to the inference as \emph{non-degenerate}.

\begin{defi}
A {\em resolution derivation}~$\mathcal D$ of a clause~$C$
from a set~$\Gamma$ of clauses
is a sequence of clauses $C_1,\ldots,C_s{=}C$
and such that each $C_i$ is either a clause from~$\Gamma$ or
is inferred by a resolution rule from two previous clauses.
If $C$ is the empty clause, $\mathcal D$ is
a {\em resolution refutation} of~$\Gamma$.
{\em Degenerate resolution} and
{\em w-resolution} derivations and refutations
are defined similarly.

The \emph{size} of a refutation
$C_1,\ldots,C_s{=}\bot$ is defined to be $s$.
\end{defi}
Derivations are typically viewed as directed acyclic graphs.
A derivation is {\em tree-like} provided its dag is a tree.
It is well known that (tree-like) resolution is sound and complete, in
that $\Gamma$ has a refutation iff it is unsatisfiable.
\begin{defi}
A refutation~$\mathcal D$ is {\em regular} provided that
no variable is used as a resolution variable
more than once
along any path in the directed acyclic graph of~$\mathcal D$.
A derivation~$\mathcal D$ of a clause~$C$ is {\em regular} provided
that, in addition, no variable appearing in~$C$ is used as
a resolution variable in~$\mathcal D$.
\end{defi}

We next define
``regular resolution derivation trees with lemmas'', or ``regRTL'',
following~\cite{BHJ:ResTreeLearning}.
The idea is that a dag-like proof can by rewritten
as a tree-like proof in which
clauses obtained earlier in the proof
can be used freely as ``learned'' lemmas.
This will be the key component in defining
Van Gelder's notion of pool proofs.

\begin{defi}
Given a tree $T$, the {\em postorder}
ordering $<_T$ of the nodes is defined as follows:
if $u,v,w$ are distinct nodes of~$T$,
$v$~is a node in the subtree rooted at the left child of~$u$,
and $w$~is a node in the subtree rooted at the right child
of~$u$, then $v<_T w<_T u$.  The {\em preorder} ordering $<^\prime_T$
is defined similarly, but stipulates that
$u<^\prime_T v<^\prime_T w$.
\end{defi}

\begin{defi}
A {\em regRTL derivation}~\cite{BHJ:ResTreeLearning}
of a clause~$C$
from a set of initial clauses~$\Gamma$
is a tree-like resolution derivation~$T$
that fulfills the following conditions:
(a)~each leaf is labeled with either a clause of~$\Gamma$ or a clause
(called a ``lemma'')
that appears earlier in~$T$ in the $<_T$ ordering;
(b)~each internal node is labeled with a clause and a literal,
and the clause is obtained by resolution
from the clauses labeling the node's children
by resolving on the given literal;
(c)~the proof tree is regular;
(d)~the root is labeled with~$C$.
If the labeling of the root is the empty
clause, $T$ is a {\em regRTL refutation}.

A {\em regWRTL derivation}~\cite{BHJ:ResTreeLearning}
is defined
similarly, but allowing
$w$-resolution inferences instead of just resolution
inferences.

A {\em pool resolution derivation}~\cite{VanGelder:PoolResolution}
is also defined
similarly, but allowing degenerate resolution inferences.
\end{defi}

\begin{prop}
\label{prop:regWRTLandPool}
If $\Gamma$ has a regWRTL refutation~$R$, then $\Gamma$
has a pool resolution refutation~$R^\prime$ with
the size of~$R^\prime$ no greater than the size of~$R$.
\end{prop}
The proof of Proposition~\ref{prop:regWRTLandPool}
is simple.  Each clause~$C$ in~$R$ corresponds to
a clause $C^\prime$ in~$R^\prime$ with $C^\prime\subseteq C$.
Arguing inductively, suppose that $C$ is derived in~$R$ from
the clauses $C_1$ and~$C_2$ using resolution literal~$x$.
Then, it is straightforward
to define $C^\prime$ from $C_1^\prime$ and $C_2^\prime$ as the unique
clause that can be inferred by degenerate resolution
from $C_1^\prime$ and $C_2^\prime$ with respect to~$x$.
\hfill \qed

The strategy of proving the existence of short
pool refutations via constructing short regWTRL refutations
is employed in the proof of Theorem~\ref{thm:pool} below.

\begin{defi} (\cite{BHJ:ResTreeLearning}).
A ``lemma'' in clause~(a) of the definition of regRTL derivations
is called an {\em input lemma} if it is derived by an {\em input}
subderivation, namely by a subderivation
in which each inference has at least one
hypothesis which is a member of~$\Gamma$ or a lemma.
A regRTI derivation is a regRTL derivation
which uses only input lemmas as lemmas.
\end{defi}

\noindent
A bit more generally, we say that a clause is ``learned''
provided it is available for use as a lemma
by virtue
of having been learned earlier in the postorder traversal
of the proof, or by
virtue of being an initial clause:

\begin{defi}
Suppose that $R$ is a regRTL (respectively, a regRTI)
refutation of~$\Gamma$,
and let $C$~be
a clause in~$R$.  The {\em learned clauses} of~$R$
at clause~$C$ are the clauses which are either
in~$\Gamma$ or which have been derived in~$R$
(respectively, have been derived by an input subderivation
in~$R$)
before~$C$ in the postordering of~$R$.
\end{defi}

Theorem~5.1 of~\cite{BHJ:ResTreeLearning} gives a
polynomial equivalence between
regRTI proofs and DPLL with clause learning without restarts.
This equivalence, however, uses non-greedy DPLL; namely,
the DPLL proof search may need to ignore contradictions
during its search.  Since most real-world DPLL search algorithms
do not ignore contradictions, and
use unit propagation whenever possible, is it natural to posit
similar properties for regRTI proofs.  These are formalized by
the next two definitions.

\begin{defi}
Let $C$ be a clause appearing in a regular derivation~$R$.
Following~\cite{BBJ:poolVsRegular2}, we write
$C^{\Tpool}$ to denote the clause containing the literals that
appear in any clause in the path from the
root of~$R$ up to and including~$C$.
\end{defi}
The clause $C^{\Tpool}$ is the same as what
\cite{BonetBuss:poolVsRegularSAT} calls~$C^+$.
The regularity
of~$R$ ensures that $C^{\Tpool}$ contains no contradictory literals.
\begin{defi}
Let $R$ be a tree-like refutation of~$\Gamma$.
A clause~$D$ in~$R$ is {\em prior-learned} for
a clause~$C$ in~$R$ if
either $D\in\Gamma$ or
there is an occurrence of~$D$ as a learned
clause which appears
in~$R$ before~$C$ in both postorder and preorder.
\end{defi}

The intuition for ``prior-learned'', is that, when reaching
the clause~$C$  while constructing~$R$
in left-to-right, depth-first order, the prior-learned
clauses are the clauses that are already available
to help derive~$C$.

\begin{defi} (See \cite{BonetBuss:poolVsRegularSAT}).
A refutation~$R$ is {\em greedy} provided that, for each clause~$C$
of~$R$, if $C$ or any subclause of~$C$~is prior-learned,
then $C$ itself is a prior-learned clause and is a leaf clause of~$R$.
A refutation~$R$ is {\em greedy and unit-propagating}
provided
that, for each clause~$C$ of~$R$, if
there is an input derivation
of some clause $C^\prime\subseteq C^\Tpool$
from the prior-learned clauses of~$R$ at~$C$
which does not resolve on any literal in~$C^\Tpool$,
then $C$~is derived in~$R$ by such a
derivation.
\end{defi}

We can now state our main results.

\begin{thm}\label{thm:pool}
The Stone principles $\Stone(G,m)$
have regWRTL refutations, and thus pool refutations, of size $O(Nm^3)$.
\end{thm}

\begin{thm}\label{thm:regRTI}
The Stone principles $\Stone(G,m)$
have regRTI refutations of size $O(N^3m^4)$.
\end{thm}

It follows from Theorem~\ref{thm:regRTI} and Theorem~5.1
of~\cite{BHJ:ResTreeLearning} that DPLL proof search with
clause learning and without restarts can refute the
Stone principle clauses in polynomial time.  It is possible
that the regRTI refutations of Theorem~\ref{thm:regRTI}
can be made greedy and unit-propagating, but we have
not tried to prove this.

The proofs of Theorems \ref{thm:pool} and~\ref{thm:regRTI}
are given in Sections \ref{sec:pool} and~\ref{sec:regrti}, respectively.
Section~\ref{sec:learning} first gives some preliminary
resolution derivations
that will be useful for both proofs.

Of course, Theorem~\ref{thm:regRTI}
implies Theorem~\ref{thm:pool} apart from the size bounds.  However,
it seems useful to prove the two theorems separately, since
the proof of Theorem~\ref{thm:pool} is substantially simpler than
the proof of Theorem~\ref{thm:regRTI}.

The intuition behind both proofs is similar.  The reason the Stone
tautologies seem highly irregular is that,
in the earlier refutations given by~\cite{AJPU:regularresolution},
some of the variables (the $r_j$'s)
are resolved on repeatedly during the refutation.  The intuition
is that the \hbox{regWRTL/regRTI} proofs for Theorems \ref{thm:pool} and~\ref{thm:regRTI}
can be
built in a bottom-up fashion starting from the empty clause,
by first resolving on the variables~$p_{i,j}$ that do
not cause irregularities, and saving the
problematic variables~$r_j$ to be resolved on later (higher in
the proof).  This is not quite completely true, since our
derivations do also resolve again on $p_{i,j}$'s at the top of
the derivations; it is nonetheless a useful intuition.

\section{Learning and 3-Learning}\label{sec:learning}

The regWRTL refutation for Theorem~\ref{thm:pool} and
the regRTI refutation for Theorem~\ref{thm:regRTI} both work
by learning the clauses $\negp_{i,j},r_j$.  If $i$ is a source
vertex of~$G$ then these clauses are Stone clauses, but
otherwise they must be learned.

Suppose that $i$ is a non-leaf vertex,
and $i^\prime$ and~$i^\pprime$ are the
two predecessors of~$i$ in~$G$.  In addition, suppose that every clause
$\negp_{i^\prime,j}, r_j$ and $\negp_{i^\pprime,j}, r_j$ has already been learned.
Fix a value of~$j$.  A derivation of $\negp_{i,j}, r_j$ proceeds in the
following three steps.

First, for each $j^\prime\not=j^\pprime$,
both distinct from~$j$,
derive the clause
\begin{equation}\label{eq:onelearnclause}
\negp_{i^\prime,j^\prime},\negp_{i^\pprime,j^\pprime},\negp_{i,j},r_j
\end{equation}
by resolving a Stone clause against the
two learned clauses $\negp_{i^\prime,j^\prime},r_{j^\prime}$
and $\negp_{i^\pprime,j^\pprime},r_{j^\pprime}$ using
$r_{j^\prime}$ and $r_{j^\pprime}$ as resolution
variables:
\begin{prooftree}
\AxiomC{$\negp_{i^\prime,j^\prime},\negr_{j^\prime},
\negp_{i^\pprime,j^\pprime},\negr_{j^\pprime},\negp_{i,j},r_j$}
\AxiomC{$\negp_{i^\prime,j^\prime},r_{j^\prime}$}
\BinaryInfC{$\negp_{i^\prime,j^\prime},
\negp_{i^\pprime,j^\pprime},\negr_{j^\pprime},\negp_{i,j},r_j$}
\AxiomC{$\negp_{i^\pprime,j^\pprime},r_{j^\pprime}$}
\BinaryInfC{$\negp_{i^\prime,j^\prime},
\negp_{i^\pprime,j^\pprime},\negp_{i,j},r_j$}
\end{prooftree}
For $j^\prime=j^\pprime\not= j$, the clause~(\ref{eq:onelearnclause})
is derived in one step by resolving
the Stone clause against only one of the two learned
clauses $\negp_{i^\prime,j^\prime},r_{j^\prime}$
and $\negp_{i^\pprime,j^\pprime},r_{j^\pprime}$.

Second, for each $j^\pprime\not= j$, resolve the
Stone clause
$\bigvee_{j^\prime=1}^m p_{i^\prime,j^\prime}$
against the learned clause~$\negp_{i^\prime,j},r_j$
and
against $m-1$ of the clauses~(\ref{eq:onelearnclause}) to
obtain
\begin{equation}\label{eq:twolearnclause}
\negp_{i^\pprime,j^\pprime},\negp_{i,j},r_j.
\end{equation}
This is shown in Figure~\ref{fig:twolearning}.

\begin{figure}[t]
\centering
\psset{unit=0.02cm}
\begin{pspicture}(-180,730)(150,1100)
\rput{0}(0,752){$\negp_{i^\pprime,j^\pprime},\negp_{i,j},r_j$}
\psline(-5,764)(-23,800)
\psline(10,764)(95,800)
\rput{0}(-35,812){$p_{i^\prime,1},\negp_{i^\pprime,j^\pprime},\negp_{i,j},r_j$}
\rput{0}(160,812){$\negp_{i^\prime,1},\negp_{i^\pprime,j^\pprime},\negp_{i,j},r_j$}
\psline(-30,824)(-48,860)
\psline(-15,824)(70,860)
\rput{0}(-65,872)%
    {$p_{i^\prime,1},p_{i^\prime,2},\negp_{i^\pprime,j^\pprime},\negp_{i,j},r_j$}
\rput{0}(130,872){$\negp_{i^\prime,2},\negp_{i^\pprime,j^\pprime},\negp_{i,j},r_j$}
\psline[linestyle=dotted,dotsep=5pt,linewidth=1.5pt](-55,882)(-82,938)
\rput{0}(-85,946) {${p_{i^\prime,1}},p_{i^\prime,2},\ldots,{p_{i^\prime,{m-2}}},\negp_{i^\pprime,j^\pprime},\negp_{i,j},r_j$}
\psline(-90,958)(-108,994)
\psline(-75,958)(55,994)
\rput{0}(125,1006){$\negp_{i^\prime,m-1},\negp_{i^\pprime,j^\pprime},\negp_{i,j},r_j$}
\rput{0}(-125,1006){${p_{i^\prime,1}},p_{i^\prime,2},\ldots,
	{p_{i^\prime,{m-1}}},\negp_{i^\pprime,j^\pprime},\negp_{i,j},r_j$}
\psline(-120,1018)(-138,1054)
\psline(-105,1018)(-20,1054)
\rput{0}(-155,1066)%
    {${p_{i^\prime,1}},p_{i^\prime,2},\ldots,
	{p_{i^\prime,{m}}}$}
\rput{0}(45,1066){$\negp_{i^\prime,m},\negp_{i^\pprime,j^\pprime},\negp_{i,j},r_j$}
\end{pspicture}
\caption{The derivation of a clause~(\ref{eq:twolearnclause})
follows this pattern with the one exception (not shown) that
the clause $\negp_{i^\prime,j},\negp_{i^\pprime,j^\pprime},\negp_{i,j},r_j$
is not one of the $m-1$ clauses~(\ref{eq:onelearnclause}) and
the learned clause $\negp_{i^\prime,j},r_j$ is used instead.}
\label{fig:twolearning}
\end{figure}
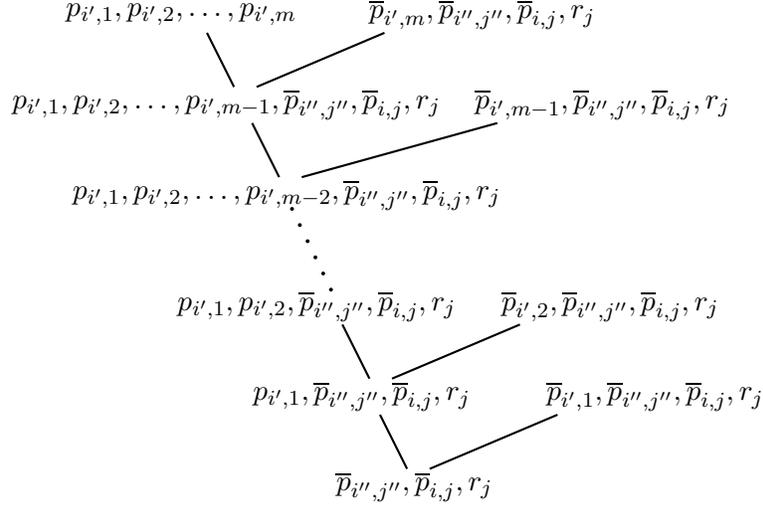

Third, resolve the Stone clause
$\bigvee_{j^\pprime=1}^m p_{i^\pprime,j^\pprime}$ against the
learned clause $\negp_{i^\pprime,j},r_j$ and against the
$m-1$~many clauses~(\ref{eq:twolearnclause}), and derive
the desired clause $\negp_{i,j}, r_j$.  This is shown in
Figure~\ref{fig:threelearning}.

\begin{figure}[t]
\centering
\psset{unit=0.02cm}
\begin{pspicture}(-150,730)(100,1100)
\rput{0}(0,752){$\negp_{i,j},r_j$}
\psline(-5,764)(-23,800)
\psline(10,764)(85,800)
\rput{0}(-25,812){$p_{i^\pprime,1},\negp_{i,j},r_j$}
\rput{0}(120,812){$\negp_{i^\pprime,1},\negp_{i,j},r_j$}
\psline(-30,824)(-48,860)
\psline(-15,824)(60,860)
\rput{0}(-55,872)%
    {$p_{i^\pprime,1},p_{i^\pprime,2},\negp_{i,j},r_j$}
\rput{0}(90,872){$\negp_{i^\pprime,2},\negp_{i,j},r_j$}
\psline[linestyle=dotted,dotsep=5pt,linewidth=1.5pt](-55,882)(-82,938)
\rput{0}(-85,946) {${p_{i^\pprime,1}},p_{i^\pprime,2},\ldots,{p_{i^\pprime,{m-2}}},\negp_{i,j},r_j$}
\psline(-90,958)(-108,994)
\psline(-75,958)(40,994)
\rput{0}(85,1006){$\negp_{i^\pprime,m-1},\negp_{i,j},r_j$}
\rput{0}(-125,1006){${p_{i^\pprime,1}},p_{i^\pprime,2},\ldots,
	{p_{i^\pprime,{m-1}}},\negp_{i,j},r_j$}
\psline(-120,1018)(-138,1054)
\psline(-105,1018)(-10,1054)
\rput{0}(-155,1066)%
    {${p_{i^\pprime,1}},p_{i^\pprime,2},\ldots,
	{p_{i^\pprime,{m}}}$}
\rput{0}(25,1066){$\negp_{i^\pprime,m},\negp_{i,j},r_j$}
\end{pspicture}
\caption{The derivation of the clause $\negp_{i,j},r_j$
follows this pattern with the one exception (not shown)
that the clause $\negp_{i^\pprime,j},\negp_{i,j},r_j$
is not one of the $m-1$ clauses~(\ref{eq:twolearnclause})
and the clause $\negp_{i^\pprime,j},r_j$
is used instead.}
\label{fig:threelearning}
\end{figure}
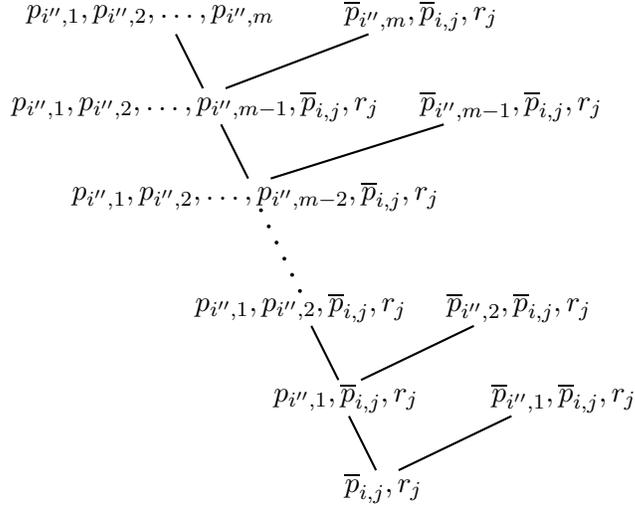

For the regWRTL proof constructed in Section~\ref{sec:pool},
the clause $\negp_{i,j},r_j$ will be learned and available to
use as a lemma once the above three steps have been carried out.

For the regRTI proof described in Section~\ref{sec:regrti},
this is not sufficient, since only clauses derived by input
subderivations are
learned.  For regRTI proofs,
the above three steps are used the first time $\negp_{i,j}, r_j$
is derived.  This results in the clauses~(\ref{eq:onelearnclause})
being learned as input lemmas, but not the clauses~(\ref{eq:twolearnclause}).
The second time $\negp_{i,j},r_j$ is derived,
only the second and third steps of the above derivation are carried out.
This results in the clauses~(\ref{eq:twolearnclause}) becoming
learned, but not the clause $\negp_{i,j},r_j$.
The third time $\negp_{i,j},r_j$ is derived, only the third step
is needed; this results in the clause~$\negp_{i,j},r_j$ becoming learned
as an input lemma.

This leads
to the following definition, which will be useful for the
regRTI derivations of Section~\ref{sec:regrti}:
\begin{defi}
Let $i,i^\prime,i^\pprime,j$ be as above, so in particular
all the clauses $\negp_{i^\prime,j^\prime},r_{j^\prime}$
and $\negp_{i^\pprime,j^\pprime},r_{j^\pprime}$ are
learned.
The clause $\negp_{i,j},r_j$ is called {\em 3-learned} provided it has
been learned.  It is called {\em 2-learned} if
all of the clauses~(\ref{eq:twolearnclause}) for $j^\pprime \neq j$ have been learned.
It is called {\em $1$-learned}
if all of the clauses~(\ref{eq:onelearnclause}) for $j^\prime \neq j$, $j^\pprime \neq j$ have been learned.

A vertex~$i$ is defined to be {$K$-learned}, for $K=1,2,3$,
if and only if every $\negp_{i,j},r_j$ has been $K$-learned.
It is also allowed that $K=0$:
every clause $\negp_{i,j},r_j$ and every vertex~$i$ is considered
to be 0-learned.
\end{defi}

Since axiom clauses are considered to be learned, the
source vertices $i > n$ are 3-learned by definition.

The next theorem summarizes the above construction.
\begin{thm}\label{thm:LearnClause}
Let $i^\prime$ and $i^\pprime$ be the two predecessors of
vertex~$i$.
There is a regular tree-like derivation of the clause $\negp_{i,j},r_j$
from Stone clauses and the clauses $\negp_{i^\prime,j^\prime},r_{j^\prime}$
and $\negp_{i^\pprime,j^\pprime},r_{j^\pprime}$, which has size $O(m^2)$
and resolves on (only) the variables
$r_k$ for $k\not=j$ and the
variables $p_{i^\prime,k}$ and $p_{i^\pprime,k}$ for all~$k$.

In the setting of a regRTI proof, if $i^\prime$ and $i^\pprime$
are 3-learned, and $\negp_{i,j},r_j$ was already \hbox{$K$-learned} for $K<3$,
then there is a regular tree-like derivation of $\negp_{i,j},r_j$
from learned clauses (including Stone clauses)
which causes it to become $(K{+}1)$-learned.
This derivation has size $O(m^2)$ and resolves on
at most the variables
$r_k$ for $k\not=j$ and
variables $p_{i^\prime,k}$ and $p_{i^\pprime,k}$ for all~$k$.
\end{thm}

It will also be useful to modify the derivations described above
to allow side variables $\negr_{j_1},\ldots,\negr_{j_\ell}$.
This is summarized by the next theorem.
\begin{thm}\label{thm:LearnClauseSide}
Let $i, i^\prime, i^\pprime$ be as above.
Let $F = \{\negr_{j_1},\ldots,\negr_{j_\ell}\}$ where $\negr_j\notin F$.
There is a regular tree-like derivation of the clause
$F,\negp_{i,j},r_j$
from Stone clauses and the clauses $\negp_{i^\prime,j^\prime},r_{j^\prime}$
and $\negp_{i^\pprime,j^\pprime},r_{j^\pprime}$, which has size $O(m^2)$
and resolves on (only) the variables
$r_k$ for $k\notin\{j,j_1,\ldots,j_\ell\}$
and the variables $p_{i^\prime,k}$ and $p_{i^\pprime,k}$
for all~$k$.
\end{thm}
The derivation for Theorem~\ref{thm:LearnClauseSide} is
obtained from the derivation for Theorem~\ref{thm:LearnClause}
by omitting
inferences that resolve on the literals~$r_{j_q}$
against the clauses $\negp_{i^\prime,j_q},r_{j_q}$
and $\negp_{i^\pprime,j_q},r_{j_q}$.

\section{The pool/regWRTL refutation}\label{sec:pool}

This section proves Theorem~\ref{thm:pool} by describing
regWRTL proofs of the Stone principles.  Fix an instance of
the Stone principle for a dag~$G$ as above with
$N$ vertices and $m$ stones.  Recall that $G$ has
$n<N$~non-source vertices.

The regWRTL refutation of $\Stone(G,m)$
will be a tree with its final, empty, clause
at the bottom.
The main part of the regWRTL refutation
above the empty clause is a ``skeleton'',
which consists of a long branch containing $n$ segments of length~$m$ each,
as is shown in Figures \ref{fig:skeleton} and~\ref{fig:skeletondetail}.
Each segment in the skeleton corresponds to a non-source vertex in $G$,
and the role of the $i$-th segment is that
clauses of the form $\negp_{i,j},r_j$ are
learned on branches to the right of the segment.
In keeping with the intuition discussed in the introduction, this
skeleton is the bottom part of the proof which resolves on
the literals $p_{i,j}$; the variables~$r_j$
(plus additional variables~$p_{i^\prime,j^\prime}$
with $i^\prime>i$) will be resolved on
above the skeleton.

The first, second, and last segments of the skeleton are
pictured in Figure~\ref{fig:skeleton}.
Each of these three segments is somewhat atypical, but the $i$-th segment for a typical intermediate~$i\in \{3,\ldots,n{-}1\}$
is pictured in Figure~\ref{fig:skeletondetail}.
In the typical situation,
the idea is that for given $3\le i\le n{-}1$ and~$j<m$,
the clause $\negp_{i,j},r_j$ is learned in the $\proofdots$ part of the proof
above the clause $\negp_{1,m},\negp_{i-1,m},\negp_{i,j}$
on the right hand side of Figure~\ref{fig:skeletondetail}.
For $i=1,2,n$ and $j\not= m$, the idea is that the
clause $\negp_{i,j},r_j$ is learned in the $\proofdots$ part
of the proof
above the clause containing  $\negp_{i,j}$.
There are various exceptions to this idea, and several
complications, as discussed below.
However, in all cases, when working in the subproof
in the $\proofdots$ part of the proof above the clause
containing $\negp_{i,j}$, all clauses of the
form $\negp_{i^\prime,j^\prime},r_{j^\prime}$ with $i^\prime>i$
have already been learned.
To maintain the regularity property, this subproof must itself
be regular,
and will not resolve on any literals $p_{i^\prime,j^\prime}$
with $i^\prime\le i$.

\begin{figure}
\centering
\psset{unit=0.017cm}
\begin{pspicture}(-150,0)(100,1100)
\rput{0}(0,0){$\bot$}
\psline(-5,12)(-23,48)
\psline(10,12)(95,48)
\rput{0}(-25,60){$p_{1,1}$}
\rput{0}(110,60){$\negp_{1,1}$}
\rput{0}(110,90){$\proofdots$}
\psline(-30,72)(-48,108)
\psline(-15,72)(70,108)
\rput{0}(-55,120)%
    {$p_{1,1},p_{1,2}$}
\rput{0}(85,120){$\negp_{1,2}$}
\rput{0}(85,150){$\proofdots$}
\psline[linestyle=dotted,dotsep=5pt,linewidth=1.5pt](-55,132)(-82,186)
\rput{0}(-85,194)
	{${p_{1,1}},p_{1,2},\ldots,{p_{1,{m-2}}}$}
\psline(-90,206)(-108,242)
\psline(-75,206)(10,242)
\rput{0}(40,254){$\negp_{1,m-1}$}
\rput{0}(40,284){$\proofdots$}
\rput{0}(-130,254){${p_{1,1}},p_{1,2},\ldots,
	{p_{1,{m-1}}}$}
\psline(-120,266)(-138,302)
\psline(-105,266)(-20,302)
\rput{0}(-165,314)%
    {${p_{1,1}},p_{1,2},\ldots,
	{p_{1,{m}}}$}
\rput{0}(0,316){$\negp_{1,m}$}
\psline(-5,328)(-23,364)
\psline(10,328)(95,364)
\rput{0}(-25,376){$\negp_{1,m},p_{2,1}$}
\rput{0}(125,376){$\negp_{1,m},\negp_{2,1}$}
\rput{0}(125,406){$\proofdots$}
\psline(-30,388)(-48,424)
\psline(-15,388)(70,424)
\rput{0}(-55,436)%
    {$\negp_{1,m},p_{2,1},p_{2,2}$}
\rput{0}(100,436){$\negp_{1,m},\negp_{2,2}$}
\rput{0}(100,466){$\proofdots$}
\psline[linestyle=dotted,dotsep=5pt,linewidth=1.5pt](-55,448)(-82,502)
\rput{0}(-95,510)
	{$\negp_{1,m},{p_{2,1}},p_{2,2},\ldots,{p_{2,{m-2}}}$}
\psline(-90,522)(-108,558)
\psline(-75,522)(10,558)
\rput{0}(55,570){$\negp_{1,m},\negp_{2,m-1}$}
\rput{0}(55,600){$\proofdots$}
\rput{0}(-140,570){$\negp_{1,m},{p_{2,1}},p_{2,2},\ldots,
	{p_{2,{m-1}}}$}
\psline(-120,582)(-138,618)
\psline(-105,582)(-20,618)
\rput{0}(-165,630)%
    {${p_{2,1}},p_{2,2},\ldots,
	{p_{2,{m}}}$}
\rput{0}(15,632){$\negp_{1,m},\negp_{2,m}$}
\psline[linestyle=dotted,dotsep=7pt,linewidth=1.5pt](7.5,644)(7.5,740)
\rput{0}(0,752){$\negp_{1,m},\negp_{n-1,m}$}
\psline(-5,764)(-23,800)
\psline(10,764)(95,800)
\rput{0}(-35,812){$\negp_{1,m},\negp_{n-1,m},p_{n,1}$}
\rput{0}(165,812){$\negp_{1,m},\negp_{n-1,m},\negp_{n,1}$}
\rput{0}(165,842){$\proofdots$}
\psline(-30,824)(-48,860)
\psline(-15,824)(70,860)
\rput{0}(-70,872)%
    {$\negp_{1,m},\negp_{n-1,m},p_{n,1},p_{n,2}$}
\rput{0}(150,872){$\negp_{1,m},\negp_{n-1,m},\negp_{n,2}$}
\rput{0}(150,902){$\proofdots$}
\psline[linestyle=dotted,dotsep=5pt,linewidth=1.5pt](-55,882)(-82,938)
\rput{0}(-85,946)
	 {$\negp_{1,m},\negp_{n-1,m},{p_{n,1}},p_{n,2},\ldots,{p_{n,{m-2}}}$}
\psline(-90,958)(-108,994)
\psline(-75,958)(10,994)
\rput{0}(100,1006){$\negp_{1,m},\negp_{n-1,m},\negp_{n,m-1}$}
\rput{0}(100,1036){$\proofdots$}
\rput{0}(-150,1006){$\negp_{1,m},{p_{n,1}},p_{n,2},\ldots,
	{p_{n,{m-1}}}$}
\psline(-120,1018)(-138,1054)
\psline(-105,1018)(-20,1054)
\rput{0}(-165,1066)%
    {${p_{n,1}},p_{n,2},\ldots,{p_{n,m}}$}
\rput{0}(15,1068){$\negp_{1,m},\negp_{n,m}$}
\rput{0}(15,1098){$\proofdots$}
\end{pspicture}
\caption{The ``skeleton'' of the regWRTL proof.}
\label{fig:skeleton}
\end{figure}
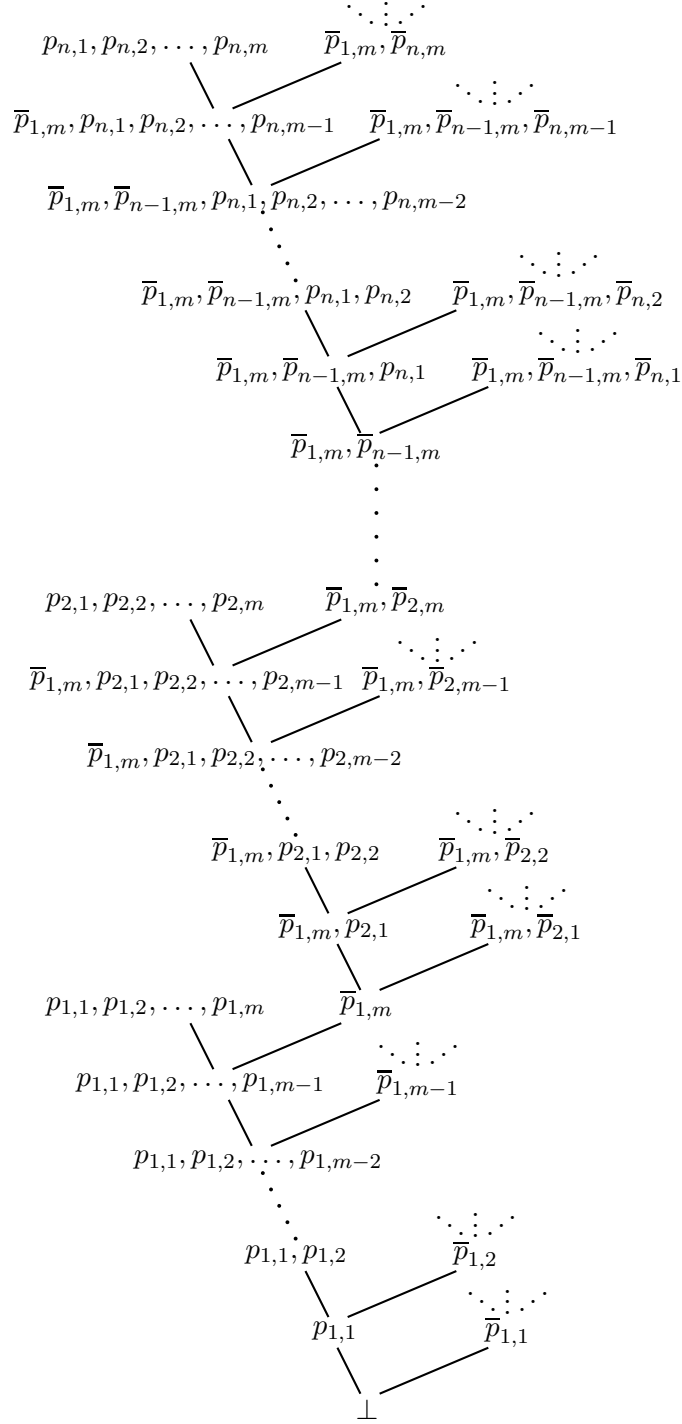

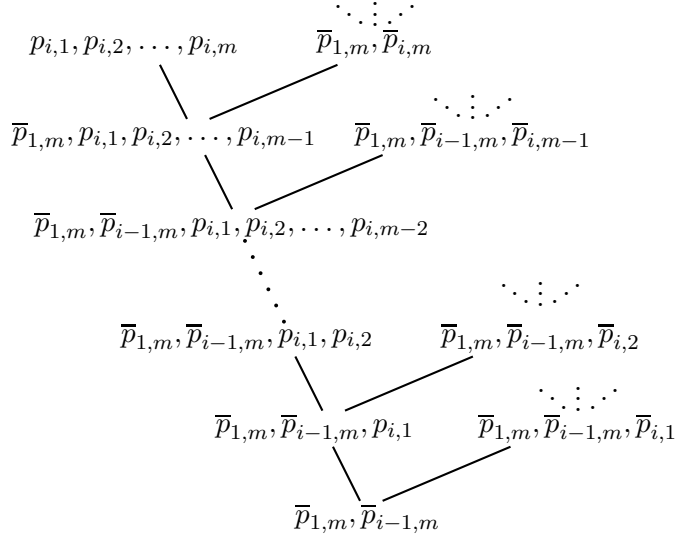
\begin{figure}[t]
\centering
\psset{unit=0.02cm}
\begin{pspicture}(-150,730)(100,1100)
\rput{0}(0,752){$\negp_{1,m},\negp_{i-1,m}$}
\psline(-5,764)(-23,800)
\psline(10,764)(95,800)
\rput{0}(-35,812){$\negp_{1,m},\negp_{i-1,m},p_{i,1}$}
\rput{0}(140,812){$\negp_{1,m},\negp_{i-1,m},\negp_{i,1}$}
\rput{0}(140,837){$\proofdots$}
\psline(-30,824)(-48,860)
\psline(-15,824)(70,860)
\rput{0}(-80,872)%
    {$\negp_{1,m},\negp_{i-1,m},p_{i,1},p_{i,2}$}
\rput{0}(115,872){$\negp_{1,m},\negp_{i-1,m},\negp_{i,2}$}
\rput{0}(115,907){$\proofdots$}
\psline[linestyle=dotted,dotsep=5pt,linewidth=1.5pt](-55,882)(-82,938)
\rput{0}(-90,946)
	 {$\negp_{1,m},\negp_{i-1,m},p_{i,1},p_{i,2},\ldots,{p_{i,{m-2}}}$}
\psline(-90,958)(-108,994)
\psline(-75,958)(10,994)
\rput{0}(70,1006){$\negp_{1,m},\negp_{i-1,m},\negp_{i,m-1}$}
\rput{0}(70,1031){$\proofdots$}
\rput{0}(-135,1006){$\negp_{1,m},p_{i,1},p_{i,2},\ldots,
	{p_{i,{m-1}}}$}
\psline(-120,1018)(-138,1054)
\psline(-105,1018)(-20,1054)
\rput{0}(-155,1066)%
    {${p_{i,1}},p_{i,2},\ldots,p_{i,m}$}
\rput{0}(5,1068){$\negp_{1,m},\negp_{i,m}$}
\rput{0}(5,1093){$\proofdots$}
\end{pspicture}
\caption{The $i$-th segment of the skeleton, for $3\le i < n$.}
\label{fig:skeletondetail}
\end{figure}

We now outline how a clause $\negp_{i,j},r_j$
is learned in the most typical case,
where $i = 3, \ldots, n{-}1$ and
$j \le m{-}2$ or $i=n$ and $j \le m{-}1$.  (The restriction
that $j<m{-}2$ when $i<n{-}1$ is made because a
more complicated construction will be needed when $j=m{-}1$
in order to also learn $\negp_{i,m},r_m$.)
The part of the proof directly above
$\negp_{1,m},\negp_{i-1,m},\negp_{i,j}$ is presented in
Figure~\ref{fig:justaboveskeleton}.

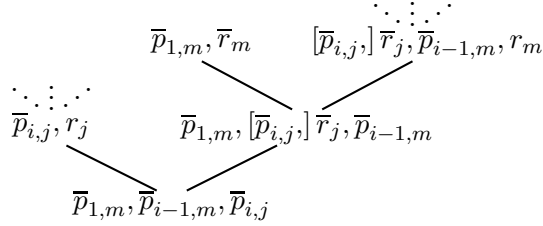
\begin{figure}[t]
\centering
\psset{unit=0.02}
\begin{pspicture}(-180,0)(220,170)
\rput{0}(0,0){$\negp_{1,m},\negp_{i-1,m},\negp_{i,j}$}
\psline(-10,10)(-70,40)
\psline(10,10)(70,40)
\rput{0}(-80,52){$\negp_{i,j},r_j$}
\rput{0}(-80,77){$\proofdots$}
\rput{0}(90,52){$\negp_{1,m},[\negp_{i,j}{,}]\,\negr_j,\negp_{i-1,m}$}
\psline(80,64)(20,96)
\psline(100,64)(160,96)
\rput{0}(20,108){$\negp_{1,m},\negr_m$}
\rput{0}(170,108){$[\negp_{i,j}{,}]\,\negr_j,\negp_{i-1,m},r_m$}
\rput{0}(160,133){$\proofdots$}
\end{pspicture}
\caption{The proof above $\negp_{1,m},\negp_{i-1,m},\negp_{i,j}$.}
\label{fig:justaboveskeleton}
\end{figure}

The clause $\negp_{i,j},r_j$ on the
left hand side of Figure~\ref{fig:justaboveskeleton}
is the clause we want to learn. This clause is derived,
and learned, by the derivation given by Theorem~\ref{thm:LearnClause}.
If $i^\prime$ and~$i^\pprime$ are the two predecessors of~$i$
in~$G$, then $i^\prime>i$ and $i^\pprime>i$ and thus
the clauses $\negp_{i^\prime,j^\prime},r_{j^\pprime}$ and~$\negp_{i^\pprime,j^\pprime},r_{j^\pprime}$ will
have all already been learned, so Theorem~\ref{thm:LearnClause}
is applicable.

On the right hand side of Figure~\ref{fig:justaboveskeleton},
we do not need to learn anything.  We only need to make sure that
the right hand side is a well-formed proof.
The notation $[\negp_{i,j},]$ indicates that $\negp_{i,j}$
may not be present.  In fact,
$\negp_{i,j}$ is present
exactly when $i$ is a predecessor of~$i{-}1$;
otherwise, it is absent.

To describe
the right hand side of Figure~\ref{fig:justaboveskeleton},
first suppose that
vertex~$i$
is not a predecessor of vertex $i{-}1$ in the graph.  In this case,
the literal $\negp_{i,j}$ is not present, and
the leaf clause $\negr_j,\negp_{i-1,m},r_m$ can be proved
by the proof given by Theorem~\ref{thm:LearnClauseSide}.
Second, suppose that vertex~$i$ is a predecessor of $i{-}1$.  We
must give a derivation of
\begin{equation}\label{eq:isPredecessor}
\negp_{i,j},\negr_j,\negp_{i-1,m},r_m .
\end{equation}
Let $i^\prime>i$ be
the other predecessor of~$i{-}1$.  The derivation
proceeds as follows. First, for each $j^\prime\notin \{j,m\}$,
it resolves
the learned clause $\negp_{i^\prime,j^\prime},r_{j^\prime}$
against the Stone clause
\[
\negp_{i^\prime,j^\prime},\negr_{j^\prime},\negp_{i,j},\negr_j,\negp_{i-1,m},r_m
\]
to obtain
\begin{equation}\label{eq:iIsPredecessorA}
\negp_{i^\prime,j^\prime},\negp_{i,j},\negr_j,\negp_{i-1,m},r_m.
\end{equation}
These steps use the resolution variables~$r_{j^\prime}$
for $j^\prime\notin \{j,m\}$.
For $j^\prime=j$, the clause~(\ref{eq:iIsPredecessorA})
is a Stone clause and does not need to be derived.
Then, it resolves the Stone clause
$\bigvee_{j^\prime} \negp_{i^\prime,j^\prime}$
against the learned clause $\negp_{i^\prime,m},r_m$ and the
$m-1$ many clauses~(\ref{eq:iIsPredecessorA}),
resolving on the literals~$p_{i^\prime,j^\prime}$.
This yields the desired clause~(\ref{eq:isPredecessor}).

That completes the description of the typical case of learning
$\negp_{i,j},r_j$. In less typical cases, the changes are as follows:
\begin{itemize}
\item $i = 2$ and $j \le m-2$. As described above, except that the variables
$p_{1,m}$ and $p_{i-1, m}$ coincide.
\item $i=1$, $j \le m-2$.  There is no $p_{i-1,m}$. The clause
$\negp_{1,j}$ is derived from
$\negp_{1,j},r_j$ and $\negp_{1,j},\negr_j$.
The former is learned using the derivation of
Theorem~\ref{thm:LearnClause}, while the latter is a Stone clause.
\item $i = n$ and $j = m$. The clause $\negp_{1,m},\negp_{n,m}$ is derived from
$\negp_{n,m},r_m$ and $\negp_{1,m},\negr_m$. The former is learned
via Theorem~\ref{thm:LearnClause}, the latter is a Stone clause.
\item $2 \ge i \ge n-1$ and $j = m-1, m$.
There is no natural place in the $i$-th segment of the
skeleton to learn clause $\negp_{i,m},r_m$,
but we must learn
it somewhere.  To create ``room''
to learn both $\negp_{i,m-1},r_{m-1}$ and $\negp_{i,m},r_m$,
the clause
$\negp_{1,m},\negp_{i-1,m},\negp_{i,m-1}$
is derived by w-resolution on
the resolution variable~$p_{i,m}$ from the two clauses $\negp_{1,m},\negp_{i-1,m},\negp_{i,m-1}$
(i.e., itself) and $\negp_{1,m},\negp_{i,m}$.
The derivation proceeds as in the typical case above the former clause
and as over $\negp_{1,m},\negp_{n,m}$ above the latter.
\item $i = 1$, $j = m-1, m$. The clause
$\negp_{1,m-1}$ is derived using a w-resolution
inference with the resolution variable~$p_{1,m}$
from itself and $\negp_{1,m}$.
Above those two
clauses, the proof proceeds as in all other cases with
$i=1$.
\end{itemize}

\noindent
We have now fully described the polynomial size refutation
for the
Stone principle
in \hbox{regWRTL}.
A size bound of $O(n m^3) = O(N m^3)$ is immediate from inspection, using
the bound of $O(m^2)$ for the size of the derivations
of Theorems \ref{thm:LearnClause} and~\ref{thm:LearnClauseSide}.
This completes the proof of Theorem~\ref{thm:pool}.

Note that our construction of the regWRTL refutation of $\Stone(G,m)$
makes no use of the assumption that $m \ge N$.  This is in
contrast to the construction of regRTI proofs in the
next section.

\section{The RegRTI proof}\label{sec:regrti}

We now give a regRTI refutation~$R$ for the $\Stone(G,m)$ principles.
We describe~$R$ by building it from
the bottom up, constructing~$R$ in a left-to-right depth-first
fashion.  At each point of the construction, $R$ is a partially
formed regRTI refutation.  The leaves of~$R$ are designated as either
``finished'' or ``unfinished'', and all finished leaves
are to the left of all unfinished leaves.  The finished leaves
are learned clauses; namely, they are
either
valid Stone clauses or have been
derived by an input subderivation earlier in the postorder of~$R$.
Clauses $\negp_{i,j}, r_j$, or vertices~$i$
are defined to be $K$-learned, $K=0,1,2,3$,
according to whether they have been $K$-learned in~$R$
at the point of reaching the leftmost unfinished leaf of~$R$.

Each unfinished leaf will contain a clause~$C$ of
the form
\begin{equation}\label{eq:unfinished}
\negp_{i_1,j_1}, \negp_{i_2,j_2}, \ldots, \negp_{i_k,j_k},
\end{equation}
for $k\ge 0$.  The \emph{domain}, $\Tdom(C)$, of~$C$
is equal to $\{i_1,\ldots,i_k\}$.
Clauses of the form~(\ref{eq:unfinished}) will be required
to have distinct values for the~$i_\ell$'s;
for convenience we assume
that
\begin{equation}\label{eq:unfinishedDomain}
i_1 ~<~ i_2 ~<~ \cdots ~<~ i_k .
\end{equation}
We let $\Tmaxdom(C)$ denote the maximum
member of $\Tdom(C)$, namely~$i_k$.
We let $\Tdom(C^\Tpool)$ denote
the set of~$i$ such that some $\negp_{i,j}\in C^\Tpool$.

We define the notion of a ``well-formed unfinished
clause'' momentarily.
The intuition behind an unfinished leaf~$C$ is based on
the idea that a proof is being constructed from the bottom up,
by a search process which sets the values of
resolution variables in such a way that the clause reached at a given point
becomes false.
Upon reaching~$C$, the search process has
set all of the literals in~$C$ (and~$C^\Tpool$) false, so that
each vertex~$i_\ell\in\Tdom(C)$ has been pebbled with
stone~$j_\ell$.  The search process will generate a
derivation of~$C$ by proving that each such stone~$j_\ell$
is red, namely that $r_{j_\ell}$ is true.  Since well-formed
initial clauses will have $i_1=1$, this will
yield a contradiction and thereby the desired refutation.

The stone $j_\ell$ can be shown to be red using
one of the following three scenarios:
(1)~$i_\ell$ is a source vertex in~$G$,
(2)~there is $i_{\ell^\prime}>i_\ell$ such that
$j_{\ell^\prime} = j_\ell$ and stone $j_{\ell^\prime}$~is red, or
(3)~the two predecessors
of~$i_\ell$ in~$G$ are pebbled with red stones.
Accordingly, for a fixed clause~$C$
of the form~(\ref{eq:unfinished})
satisfying~(\ref{eq:unfinishedDomain}), we define:
\begin{defi}
Let $1\le \ell \le k$. The vertex~$i_\ell$ is said to be \emph{bypassed}
if there is
some $\ell^\prime >\ell$ such that
$j_{\ell^\prime} = j_\ell$.  For the maximum such value~$\ell^\prime$,
the vertex~$i_{\ell^\prime}$
is called the \emph{max-bypasser} of~$i_\ell$.

Now let $1\le \ell < \ell^\prime \le k$.  We say that
vertex~$i_{\ell^\prime}$ \emph{directly supports} vertex~$i_\ell$
if either (1)~$i_{\ell^\prime}$~is the max-bypasser of~$i_\ell$,
or (2)~$i_\ell$~is not bypassed and $i_{\ell^\prime}$ is
one of the two predecessors of~$i_\ell$ in~$G$.
Note that $i_\ell$ can
directly support multiple~$i_{\ell^\prime}$'s.

The ``\emph{supports}'' relation is the reflexive, transitive
closure of ``directly supports''; namely, if
$i_{\ell_1} > i_{\ell_2} > \cdots > i_{\ell_s}$, $s\ge 1$,
and each $i_{\ell_q}$ directly supports $i_{\ell_{q-1}}$,
then $i_{\ell_1}$ supports~$i_{\ell_s}$.  We use $S_\ell(C)$
to denote the set of
vertices in~$\Tdom(C)$ which support~$i_\ell$.
Note that $S_\ell(C) \subseteq \{i_\ell,\ldots,i_k\}$.
\end{defi}

The construction of~$R$ starts with the empty clause as
the first unfinished clause.  The first step will be
to generate new unfinished
clauses of the form $\negp_{1,j_1}$;
namely of the form~(\ref{eq:unfinished}) with
$k=1$ and $i_1=1$.  In subsequent steps, an unfinished
clause is extended by adding literals $\negp_{i_{k+1},j_{k+1}}$
where $i_{k+1}$ is the least vertex $>i_k$ which supports
the vertex $i_1=1$.  The next definitions make this formal.
\begin{defi}
A clause~$C$ is a \emph{well-formed unfinished clause}
provided $C$ is of the form~(\ref{eq:unfinished}),
satisfies (\ref{eq:unfinishedDomain}) with $i_1=1$
and $i_k\le n$,
and the following three conditions hold:
\begin{description}
\item[\rm i] $C^\Tpool$ contains only literals of the
form~$\negp_{i,j}$ for $i\le i_k$, and for each~$i$ there is at most
one~$j$ such that $\negp_{i,j}$ in $C^\Tpool$.
\item[\rm ii] Let $1\le \ell\le k$, and consider $i_\ell$.  Then either
\begin{description}
\item[\rm a] The vertex $i_\ell$ is bypassed, or
\item[\rm b] The vertex $i_\ell$ is not bypassed, and
each predecessor~$i^\prime$ of~$i_\ell$ in~$G$ satisfies
one of the following three conditions:
\begin{description}
\item[\rm ($\alpha$)]
$i^\prime\in \Tdom(C)$;
\item[\rm ($\beta$)]
Vertex~$i^\prime$ is already 3-learned,
and $i^\prime\notin\Tdom(C^\Tpool)$; or
\item[\rm ($\gamma$)] $i^\prime>i_k$,
and $i^\prime$~is not 3-learned.
\end{description}
In case~($\alpha$), $i^\prime$ may
or may not be 3-learned.
\end{description}
\item[\rm iii] Each $i_\ell\in\Tdom(C)$ supports
$i_1=1$. Equivalently, $S_1(C) = \Tdom(C)$.
\end{description}
The empty clause is also a well-formed unfinished clause.
\end{defi}

\begin{defi}
A non-empty well-formed unfinished clause of the form~(\ref{eq:unfinished})
is {\em extendible} provided there is some $i_\ell$
with at least one predecessor~$i^\prime$ in~$G$ that
satisfies condition~($\gamma$).  The empty clause
is also extendible.
\end{defi}

During the construction of~$R$, all unfinished leaves will be
well-formed unfinished clauses. To describe the construction, we explain how to
handle the leftmost unfinished leaf of the so-far constructed
portion of~$R$.

\subsection*{The extendible case}
First suppose that the leftmost unfinished leaf is an extendible
clause~$C$ of the form~(\ref{eq:unfinished})
with $k>0$.
Considering all vertices~$i_\ell$,
find one with
the least predecessor~$i^\prime$ that satisfies~($\gamma$).
This least~$i^\prime$ is
denoted~$i_{k+1}$.
In the special case where $k=0$ and $R$~contains just the empty clause (as~$C$),
let $i_{k+1} = 1$.
In either case, $i_{k+1}$ is not 3-learned
and not a source vertex for~$G$, so $i_{k+1}\le n$.

With $i_{k+1}$ chosen, define $D_t$ to be
the clause $C,\negp_{i_{k+1},t}$.  The idea
is that we would like to replace $C$ in~$R$ with
a derivation of~$C$ from the Stone clause
$p_{i_{k+1},1},\ldots,p_{i_{k+1},m}$ and
the clauses~$D_t$
for $t=1,\ldots, m$ by resolving on the variables
$p_{i_{k+1},t}$.  The problem with this is that
the $D_t$'s may not be well-formed unfinished clauses.
So, we consider the set $S_1(D_t)$, namely
the set of literals that support the root vertex~$1$
in~$D_t$.
\begin{clm}\label{cl:containst}
$i_{k+1} \in S_1(D_t)$.
\end{clm}
The claim is trivial for $k=0$.
To prove it when $k>0$, first suppose that
$t$ is equal to some $j_\ell$ for $\ell\le k$.
Choose $i_\ell$ to be the
least value such that $t=j_\ell$; of course
$i_{k+1}$ is a max-bypasser for~$i_\ell$ in~$D_t$.
Using the fact that $S_1(C)$ contains
every $i_{\ell^\prime}$ for $\ell^\prime\le k$,
a simple induction argument proves
that $i_{\ell^\prime}\in S_1(D_t)$ for every $\ell^\prime\le \ell$.
It follows that $i_{k+1}\in S_1(D_t)$ since it is the max-bypasser
of $i_\ell\in S_1(D_t)$.
Second, suppose that $t$ is distinct from all the $j_\ell$ values.
A similar induction argument proves readily
that $S_1(D_t)$ contains every $i_\ell$ for $\ell \le k$.
Hence $i_{k+1}\in S_1(D_t)$ since $i_{k+1}$ is a predecessor of some
non-bypassed $i_\ell\in S_1(D_t)$.
\hfill \qed
\begin{clm}\label{cl:unionok}
There is a $t$ such that $S_1(D_t) = \{i_1,\ldots,i_k,i_{k+1}\}$.
\end{clm}
To prove this, take $t$ distinct from all $j_\ell$ values,
and use the result from the second subcase of the previous claim.
There must exist such a~$t$ since $m\ge N$, i.e., there are at least
as many stones as vertices.
\hfill \qed

\medskip

\begin{figure}[t]
\centering
\psset{unit=0.02cm}
\begin{pspicture}(-180,750)(100,1080)
\rput{0}(0,752){$C$}
\psline(-5,764)(-23,800)
\psline(10,764)(100,800)
\rput{0}(-25,812){$C^*_{m-1},p_{i_{k+1},m}$}
\rput{0}(135,812){$C_m,\negp_{i_{k+1},m}$}
\psline(-34,824)(-52,860)
\psline(-19,824)(71,860)
\rput{0}(-60,872){$C^*_{m-2},p_{i_{k+1},m-1},p_{i_{k+1},m}$}
\rput{0}(125,872){$C_{m-1},\negp_{i_{k+1},m-1}$}
\psline[linestyle=dotted,dotsep=5pt,linewidth=1.5pt](-62,887)(-85,943)
\rput{0}(-85,956){$C^*_2,p_{i_{k+1},3},\ldots,p_{i_{k+1},m}$}
\psline(-95,968)(-113,1004)
\psline(-82,968)(3,1004)
\rput{0}(30,1016){$C_2,\negp_{i_{k+1},2}$}
\rput{0}(-125,1016){$C_1,p_{i_{k+1},2},\ldots,p_{i_{k+1},m}$}
\psline(-123,1028)(-141,1064)
\psline(-108,1028)(-23,1064)
\rput{0}(-145,1076){$p_{i_{k+1},1},\ldots,p_{i_{k+1},m}$}
\rput{0}(10,1076){$C_1,\negp_{i_{k+1},1}$}
\end{pspicture}
\caption{The proof in the extendible case.  The resolution variables
are the $p_{i_{k+1},s}$'s. Each clause $C_i^*$ is
the union of the clauses $C_1,\ldots,C_i$.}
\label{fig:extendible}
\end{figure}
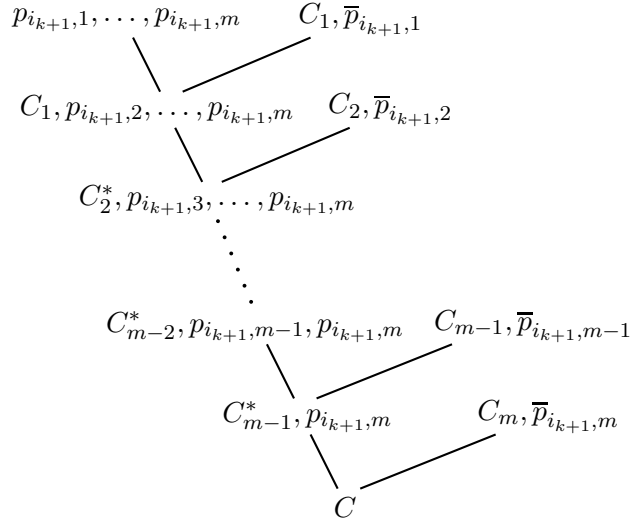

Define $C_t$ to be the clause containing the literals
$\negp_{i_\ell,j_\ell}$ for $i_\ell\in S_1(D_t)$ with $1\le \ell \le k$.
Claim~\ref{cl:unionok} shows that $C = \bigcup_t C_t$.
With the aid of Claim~\ref{cl:containst}
and the fact that $C$ satisfies conditions ii and~iii of the
definition of well-formedness, it follows
from the definition of~$C_t$ that
the clause $C_t,\negp_{i_{k+1},t}$
also satisfies the conditions ii and~iii.

Now replace the clause~$C$ in~$R$ with the
resolution derivation shown in Figure~\ref{fig:extendible}.
The clauses~$C_t,\negp_{i_{k+1},t}$
are clearly well-formed unfinished clauses,
the $C_t$'s are subclauses of~$C$, and the
(sub)clauses~$C^*_j$ in the figure are defined
to equal $\bigcup_{t\le j} C_t$.
By Claim~\ref{cl:unionok},
$C^*_m$ is equal to~$C$, which implies that the
inference used to derive $C$ is a valid non-degenerate resolution
inference. For all the other newly
added inferences, this is clear.

That completes the construction
in the case where the leftmost unfinished leaf is extendible.

\subsection*{The non-extendible case}
The construction for the case where
the leftmost unfinished leaf is not extendible
is summarized in the following lemma.

\begin{lem}\label{lm:closure}
Suppose that clause~$C$, of the form~{\rm (\ref{eq:unfinished})},
is the leftmost unfinished leaf of~$R$, and that $C$~is not
extendible.  Then there is a regRTI derivation~$R_C$ of~$C$ from
the Stone clauses and the learned clauses in~$R$
to the left of~$C$.
If the clause $\negp_{i_k,j_k},r_{j_k}$
was already $K$-learned for $K<3$, then it becomes $(K{+}1)$-learned
after the clause~$C$ is replaced in~$R$ with~$R_C$.
The only variables used as resolution variables in~$R_C$
are variables~$r_j$, and variables $p_{i,j}$ for $i\notin\Tdom(C^\Tpool)$.
The size of~$R_C$ is $O(Nm^2)$.
\end{lem}
\begin{proof}  First suppose that $k=1$.  Then $i_1=1$,
so $\negp_{i_1,j_1},\negr_{i_1}$ is a Stone clause.
Since $C$ is not extendible,
the two predecessors $i^\prime$ and~$i^\pprime$ of~$i_1$
are 3-learned.  Theorem~\ref{thm:LearnClause} gives
a resolution derivation~$L_1$ of $\negp_{i_1,j_1},r_{j_1}$
from 3-learned clauses using
resolution on (at most) the variables $r_j$ with $j\not=j_1$
and the variables $p_{i^\prime,j}$ and~$p_{i^\pprime,j}$.
Form the
derivation~$R_C$ as
\begin{equation}\label{eq:rjklearning}
\hbox{\AxiomC{$\proofdots^{\textstyle L_1}$}
\noLine
\UnaryInfC{$\negp_{i_1,j_1},r_{j_1}$}
\AxiomC{$\negp_{i_1,j_1},\negr_{j_1}$}
\BinaryInfC{$\negp_{i_1,j_1}$}
\DisplayProof }
\end{equation}
Here and in the sequel, notation of the form $\proofdots^{\textstyle L}$ indicates that
the dots should be replaced by the derivation $L$ without its (already pictured) final
clause.

By the second half of Theorem~\ref{thm:LearnClause},
if $\negp_{i_1,j_1},r_{j_1}$ was $K$-learned for $K<3$,
it becomes $(K{+}1)$-learned.
The size of~$L_1$ is $O(m^2)$.

We now assume that $k > 1$. This case is more involved, and takes
up most of the rest of the paper.

Define $B$ to be the set of those non-bypassed vertices $i_\ell \in \Tdom(C)$ that
have at least one predecessor outside of $\Tdom(C)$.
Note that since $C$ is not extendible,
a non-bypassed element of $\Tdom(C)$ belongs to $B$ exactly if
it has a predecessor satisfying condition~($\beta$)
from the definition of well-formedness; in particular,
the predecessor has to be 3-learned. We further split $B$
into the sets $B_1$ and $B_2$ depending on whether
one or both predecessors satisfy~($\beta$). Note that
$i_1=1$~might be in~$B_1$, but it cannot be in~$B_2$,
because  $i_2$ is either a predecessor
or a max-bypasser of~$i_1$.
On the other hand, $i_k$~must be in~$B_2$.
We also define $B^+$ to be the
set containing~$B$ and all those $i_\ell$
which are bypassed by a max-bypasser~$i_{\ell^\prime} \in B$.
Thus, a vertex $i_\ell$ is in $B^+$ if, according
to the minimal partial assignment falsifying clause~$C$,
$i_\ell$~is pebbled by a stone that also covers
some vertex in~$B$.

The idea behind the construction of~$R_C$ is that
vertices $i_\ell$ that belong to
$B \cup \{1\}$ determine a ``partition''
of $G{\upharpoonright_{\Tdom(C)}}$
into possibly overlapping subgraphs (namely the sets $S_\ell^B(C)$
defined next). $R_C$ will deal with these subgraphs independently.
The set $S_\ell^B(C)$ is a subgraph with sink~$i_\ell$:

% \begin{defi}
% The set $S_\ell^B(C)$ is the smallest set containing $i_\ell$
% such that, if $i_{\ell^\prime}\in S_\ell^B(C)$,
% $i_{\ell^\prime}$ is directly
% supported by~$i_{\ell^\pprime}$,
% and
% either $\ell^\prime=\ell$ or $i_{\ell^\prime}\notin B^+$,
% then $i_{\ell^\pprime}\in S_\ell^B(C)$.
% \end{defi}

\begin{defi}
For $i_\ell \in \Tdom(C)$, the set $S_\ell^B(C)$ is the smallest set containing $i_\ell$
and satisfying the following
whenever $i_{\ell^\prime}\in S_\ell^B(C)$ and
either $\ell^\prime=\ell$ or $i_{\ell^\prime}\notin B$:
(1)~If $i_{\ell^\prime}$ is bypassed in~$C$ by
max-bypasser~$i_{\ell^\pprime}\notin B$,
then $i_{\ell^\pprime}\in S_\ell^B(C)$
and
(2)~if $i_{\ell^\prime}$ is not bypassed in~$C$ and
$i_{\ell^\pprime}$ is a predecessor of~$i_{\ell^\prime}$ in~$G$,
then $i_{\ell^\pprime}\in S_\ell^B(C)$.
\end{defi}
Note that the closure condition~(1)
does not allow $i_{\ell^\pprime}\in B$,
whereas (2)
does allow it.

Enumerate $B\cup\{1\}$ in decreasing order
as $i_{t_1},i_{t_2},\ldots,i_{t_r}$ so that $t_1=k$ and $t_r=1$.
Below, we use the convention that
the index~$q$ ranges over $1,\ldots,r$, so that the
vertices~$i_{t_q}$ are the members of $B\cup\{1\}$.

The overall structure of the derivation~$R_C$ is shown
in Figure~\ref{fig:nonextendible} below.
$R_C$ will be built around
certain clauses related to the vertices $i_{t_q} \in B \cup \{1\}$
(the clauses $C^*_{q-1},F_q$ in Figure~\ref{fig:nonextendible}).
To complete~$R_C$, we will have to construct
derivations of the clauses to the side of this ``skeleton''
(the derivations~$L_{t_q}$
of the clauses $C_{t_q},F_q,r_{j_{t_q}}$ in Figure~\ref{fig:nonextendible}).
In the case where $i_{t_q} \in B_2$, this is easy (Lemma~\ref{lm:LtqB2} below),
and this easy case is the one which makes $i_{t_1} = i_k$
become $(K{+}1)$-learned instead of $K$-learned.
In the case where $i_{t_q} \in B_1 \cup \{1\}$, we only need to obtain
a valid derivation, but this requires a relatively complex construction
based on the structure of the subgraph determined by $i_{t_q}$
(Lemma~\ref{lm:Ltq}).

A clause appearing in~$R_C$ will contain
some literals of the form~$\negp_{i,j}$, some
literals of the form~$\negr_{j}$,
and at most one literal of the form~$r_j$.
These literals have to be selected
so as to avoid irregularities and degenerate inferences.
$R_C$~is defined using
four special types of clauses: $C_\ell$, $C_q^*$, $E_\ell$, and~$F_q$.
(The $C_\ell$'s are different from the $C_t$'s
used for the extendible case.)
The $C_\ell$'s and~$C_q^*$'s consist of~$\negp_{i,j}$'s,
while the $E_\ell$'s and~$F_q$'s consist of~$\negr_{j}$'s.
As suggested by the notation,
the $C^*_q$'s and~$F_q$'s are parametrized
by vertices $i_{t_q} \in B \cup \{1\}$,
whereas the $C_\ell$'s and~$E_\ell$'s
are parametrized by vertices $i_\ell \in \Tdom(C)$.

%$C_\ell$~is the
%subclause of~$C$ containing those literals $\negp_{i_{\ell^\prime},j_{\ell^\prime}}$
%where $i_{\ell^\prime}$ ``B-supports'' vertex~$i_\ell$,
%but omitting~$\negp_{i_\ell,j_\ell}$ if $i_\ell$ is bypassed:
\begin{defi}
If $i_\ell$ is not bypassed in the clause~$C$,
then $C_\ell$ is the
clause containing the
literals $\negp_{i_{\ell^\prime},j_{\ell^\prime}}$
for $i_{\ell^\prime}\in S_\ell^B(C)$.  If $i_\ell$ is bypassed in~$C$,
then $C_\ell$ is the same set
except the literal $\negp_{i_\ell,j_\ell}$ is omitted.
\end{defi}

\begin{defi}
$C_q^*$ equals
$\{\negp_{1,j_1}\}\cup\bigcup_{q^\prime> q}C_{t_{q^\prime}}$.
\end{defi}
The reason we have to explicitly include
$\negp_{1,j_1}$ in $C_q^*$ is that
$\negp_{1,j_1}\notin C_1$ if vertex~$1$ is bypassed.

\begin{defi}
If $i_\ell\in B$, then
the clause $E_\ell$ is the set of literals~$\negr_{j_{\ell^\prime}}$
such that $\ell^\prime\not=\ell$ and $i_{\ell^\prime}\in S_\ell^B(C) \cap B^+$.
For $i_\ell\notin B$,
the clause~$E_\ell$ contains the literals~$\negr_{j_{\ell^\prime}}$
such that $i_{\ell^\prime}\in S_\ell^B(C)\cap B^+$.
\end{defi}
Informally, $E_\ell$ contains
the literals~$\negr_{j_{\ell^\prime}}$
for $i_{\ell^\prime} \in B^+$ a source (leaf) vertex
above~$i_\ell$ relative to the subgraph~$S_\ell^B(C)$.

The cases where $i_\ell$ is bypassed by a max-bypasser~$i_{\ell^\prime}$
deserve special mention.  Of course, $j_\ell = j_{\ell^\prime}$.
If $i_{\ell^\prime}\in B$,
then $i_\ell\in B^+\setminus B$, and we have
$S_\ell^B(C) =\{ i_\ell \}$, $C_\ell = \emptyset$, and
$E_\ell = \{ \negr_{j_\ell} \}$.
On the other hand, if $i_{\ell^\prime}\notin B$,
then $i_\ell\notin B^+$, and we have
$S^B_\ell(C) = \{i_\ell\} \cup S^B_{\ell^\prime}(C)$,
$C_\ell = C_{\ell^\prime}$,
and $E_\ell = E_{\ell^\prime}$.

Before defining the $F_q$'s, we prove a lemma listing some
basic properties of the $C$'s, $C^*$'s, and $E$'s:

\begin{lem}\label{lm:SimpleECS} \leavevmode
\begin{description}
\item[\rm (a)] Each
$i_\ell\not= 1$ is a member of some $S_{t_q}^B(C)$.
Thus, $C_0^*$ is equal to~$C$.
\item[\rm (b)] If $1$ is bypassed by max-bypasser~$i_\ell$,
then $\ell=2$.
\item[\rm (c)] $\negr_{j_{t_{r-1}}} \in E_1$.
\item[\rm (d)] $E_{t_q} \subseteq \{ \negr_{j_{t_1}},\ldots,\negr_{j_{t_{q-1}}} \}$, for all~$q$.
\end{description}
\end{lem}
\begin{proof}
(a)~Given $\ell\not=1$, consider a chain
$i_{\ell}{=}i_{\ell_1} > i_{\ell_2} > \cdots > i_{\ell_s} {=} 1$
of directly supporting vertices from $i_\ell$ to~$1$.
Let $i_{\ell_a}$ be the first member of~$B\cup\{1\}$ in this sequence.
Then $\ell_a = t_{q^\prime}$ for some~$q^\prime$, and
we have {$i_\ell \in S_{t_{q^\prime}}^B(C)$}.
Thus, $\negp_{i_\ell,j_\ell}$ is in $C_{t_{q^\prime}}$
and hence~$C_0^*$.

(b)~Suppose that $1$ has max-bypasser~$i_\ell$ with $\ell>2$.
Then $i_\ell > i_2$.
There must exist a chain of directly supporting vertices
from $i_2$ to~$1$, but this contradicts $i_\ell>i_2$.

(c)~Consider a chain
$i_{t_{r-1}}{=}i_{\ell_1} > i_{\ell_2} > \cdots > i_{\ell_s} {=} 1$
of directly supporting
vertices from $i_{t_{r-1}}$ to~$1$.
By the descending order of the~$t_q$'s,
none of $i_{\ell_2},\ldots,i_{\ell_{s-1}}$
is in~$B$.
If $i_{t_{r-1}}$ is not the max-bypasser of~$i_{\ell_2}$, then
$i_{t_{r-1}}\in S_1^B(C)$ and thus $\negr_{j_{t_{r-1}}}\in E_1$.
Suppose instead that $i_{t_{r-1}}$ is the max-bypasser of~$i_{\ell_2}$,
so $i_{\ell_2}\in B^+\setminus B$.
Then $i_{\ell_2}\in S_1^B(C)$, so $\negr_{j_{\ell_2}}\in E_1$.
By the definition of bypasser, $j_{\ell_2}=j_{t_{r-1}}$,
so $\negr_{j_{t_{r-1}}}\in E_1$.

(d) Let $\negr_{j_\ell}\in E_{t_q}$, so
$\ell\ge t_q$ and $i_{\ell} \in B^+$.
If $i_{\ell} \in B$, then $\ell>t_q$,
and thus $\ell = t_{q^\prime}$ for some $q^\prime<q$.
If $i_{\ell}\in B^+\setminus B$, then $i_\ell$ has
max-bypasser~$i_{\ell^\prime} \in B$.
Then $\ell^\prime>\ell\ge t_q$
and thus $\ell^\prime$ equals $t_{q^\prime}$ for some $q^\prime<q$, and
by the definition of bypasser $j_{\ell^\prime}=j_\ell$. Therefore
$\negr_{j_\ell}$ is the same as $\negr_{j_{t_{q^\prime}}}$.
\end{proof}

Define $F_q$ to be the clause
\[
F_q ~:=~~
\left\{ \begin{array}{ll}
E_{t_r} = E_1 \qquad & \hbox{if $q=r$ and $1\notin B$} \\
\negr_{j_{t_1}},\ldots,\negr_{j_{t_{q-1}}} \qquad & \hbox{otherwise}
\end{array}\right.
\]
By Lemma~\ref{lm:SimpleECS}(d),
$E_{t_q} \subseteq F_q$.  Note that $F_1 = \emptyset$, because $r \neq 1$
and so in evaluating $F_1$ we use the second case of the definition of $F_q$.

We can now describe the derivation~$R_C$ in detail. As mentioned,
the general structure of~$R_C$ is shown in Figure~\ref{fig:nonextendible}.
The parts of the derivation displayed in brackets are omitted if $1 \in B^+\setminus B$.

Note that $C^*_0,F_1$, the final clause of $R_C$, is the same as~$C$. For $q<r$, the inference
\begin{prooftree}
\AxiomC{$C_{t_q},F_q,r_{j_{t_q}}$}
\AxiomC{$C^*_q,F_{q+1}$}
\BinaryInfC{$C^*_{q-1},F_q$}
\end{prooftree}
resolves on~$r_{j_{t_q}}$, and it is non-degenerate
by the definitions of the $F_q$'s and $C_q^*$'s.
(By Lemma~\ref{lm:SimpleECS}(c,d),
this is true
even in the case where $q=r-1$ and $1\notin B$.)
It follows that the resolution variables on the path from
$C_{t_q},F_q,r_{j_{t_q}}$ to $C$
are exactly $r_{j_{t_1}}, \ldots, r_{j_{t_q}}$.
The derivation  $L_{t_q}$ is described below (cf.\ Lemmas \ref{lm:LtqB2} and \ref{lm:Ltq}).

For $q=r$, we have $t_q$ equal to~$1$.
If $1\notin B^+\setminus B$, then,
as shown in Figure~\ref{fig:nonextendible}, the clause
$C^*_{r-1},F_{r}$ is derived by:
\begin{prooftree}
\AxiomC{$\proofdots^{\textstyle L_1}$}
\UnaryInfC{$C_1,F_r,r_{j_1}$}
\AxiomC{$\negp_{1,j_1},\negr_{j_1}$}
\BinaryInfC{$C^*_{r-1},F_r$}
\end{prooftree}
Again, the final inference deriving $C^*_{r-1},F_{r}$ is non-degenerate.
The upper right clause $\negp_{1,j_1},\negr_{j_1}$ is a Stone clause.
The derivation $L_1$~is described in Lemma~\ref{lm:Ltq}.

On the other hand, if $1\in B^+\setminus B$
and thus 1 has max-bypasser $i_2\in B$,
then
$C_{r-1}^*$ is $\negp_{1,j_1}$
and $F_r = E_1$ is $\negr_{j_1}$.
Therefore clause $C_{r-1}^*,F_r$
is equal to the Stone clause $\negp_{1,j_1},\negr_{j_1}$,
and there is no need to add to $R_C$ anything above it.

To finish the description
of~$R_C$, we must describe the derivations $L_{t_q}$.  This
is done by the next two lemmas.

\begin{lem}\label{lm:LtqB2}
Suppose that $i_{t_q}\in B_2$.  Then there is a regRTI proof~$L_{t_q}$
of $C_{t_q},F_q,r_{j_{t_q}}$ of size~$O(m^2)$.  The variables
used as resolution variables in~$L_{t_q}$ are (at most) the
variables~$r_j$ where $\negr_j \notin F_q\cup\{\negr_{j_{t_q}}\}$
and the variables $\negp_{i^\prime,j}$ and~$\negp_{i^\pprime,j}$
for $i^\prime$ and~$i^\pprime$ the predecessors of~$i_{t_q}$ in~$G$.

In the special case of $q=1$, so $t_1=k$, if $\negp_{i_k,j_k},r_{j_k}$
was $K$-learned for $K<3$, then $\negp_{i_k,j_k},r_{j_k}$
becomes $(K{+}1)$-learned in~$L_k$.
\end{lem}
Note that the values $i^\prime$ and~$i^\pprime$
are not in $\Tdom(C^\Tpool)$
by the definition of~$B_2$; therefore, the resolution variables
$\negp_{i^\prime,j}$ and~$\negp_{i^\pprime,j}$ do not violate the
regularity condition.  Also note that, since $1\notin B_2$,
we have $t_q\not= 1$ and the
condition that $\negr_j \notin F_q\cup\{\negr_{j_{t_q}}\}$
is equivalent to $j\notin \{j_{t_1},\ldots,j_{t_q}\}$. This
is precisely what is needed to ensure regularity.
\begin{proof}
The clause $C_{t_q}$ consists of the single literal $\negp_{i_{t_q},j_{t_q}}$.
Thus the derivation~$L_{t_q}$ is the derivation given by
Theorem~\ref{thm:LearnClauseSide}.
When $q=1$ and $t_q =k$, we have $F_1 = \emptyset$
and $C_k$ is the clause $\negp_{i_k,j_k}$.  In this
case, $L_k$ is the derivation given by the second part
of Theorem~\ref{thm:LearnClause}.
\end{proof}

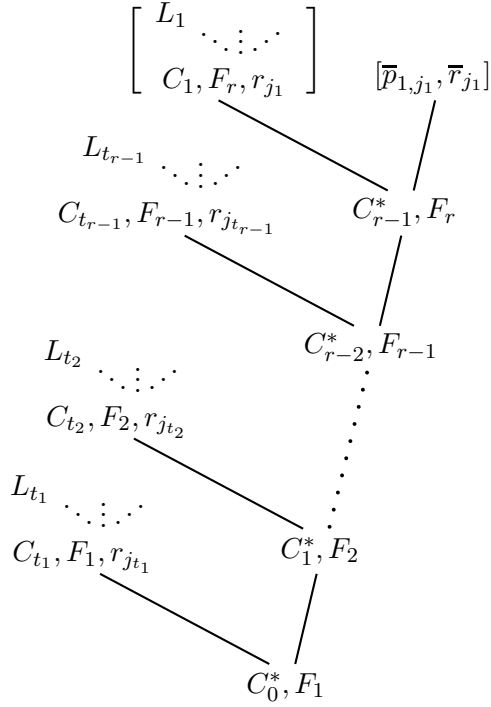
\begin{figure}[t]
\centering
\psset{unit=0.015cm}
\begin{pspicture}(-150,1180)(70,1780)

\rput{0}(0,1200){$C^*_{0},F_1$}

\psline(-15,1220)(-165,1300)
\psline(8,1220)(27,1300)

\rput{0}(-180,1344){$\begin{array}{c}{\vphantom{\proofdots}}^{\textstyle L_{t_1}}~\proofdots~~~~~~~ \\
C_{t_1},F_1, r_{j_{t_1}}\end{array}$}
\rput{0}(30,1320){$C^*_{1},F_2$}
%\rput{0}(-115, 1335){$\proofdots^{\textstyle L}$}

\psline(15,1340)(-135,1420)
% \psline(38,1340)(57,1420)

\rput{0}(-150,1464){$\begin{array}{c}{\vphantom{\proofdots}}^{\textstyle L_{t_2}}~\proofdots~~~~~~~~ \\
C_{t_2},F_2, r_{j_{t_2}}\end{array}$}
% \rput{0}(60,1440){$C^*_{2},F_3$}
%
% \psline(45,1460)(-120,1540)
\psline[linestyle=dotted,dotsep=5pt,linewidth=1.5pt](38,1340)(72,1480)
%
% \rput{0}(-120,1584){$\begin{array}{c}{\vphantom{\proofdots}}^{\textstyle L_{t_3}}~\proofdots~~~~~~~~ \\
% C_{t_3},F_3, r_{j_{t_3}}\end{array}$}
\rput{0}(75,1500){$C^*_{r-2},F_{r-1}$}

\psline(60,1520)(-90,1600)
\psline(83,1520)(102,1600)

\rput{0}(-105,1644){$\begin{array}{c}{\vphantom{\proofdots}}^{\textstyle L_{t_{r-1}}}~\proofdots~~~~~~~~~~~ \\
C_{t_{r-1}},F_{r-1}, r_{j_{t_{r-1}}}\end{array}$}
\rput{0}(105,1620){$C^*_{r-1},F_{r}$}

\psline(90,1640)(-60,1720)
\psline(113,1640)(132,1720)

\rput{0}(-55,1764){$\left[\begin{array}{c}{\vphantom{\proofdots}}^{\textstyle L_1}~\proofdots~~~~~~~ \\
C_{1},F_r, r_{j_{1}}\end{array}\right]$}
\rput{0}(130,1740){$[\negp_{1,j_1},\negr_{j_1}]$}
\end{pspicture}
\caption{The structure of the derivation $R_C$.}
\label{fig:nonextendible}
\end{figure}

\begin{lem}\label{lm:Ltq}
Suppose that $1\le q\le r$,
and $i_{t_q}\in B_1\cup\{1\}$ and $i_{t_q}\notin B^+\setminus B$.
Let $N_q = |S_{t_q}^B(C)|$.
\begin{description}
\item[\rm (a)]
There is a regular dag-like derivation~$L^\prime_{t_q}$
of size $O(N_q)$ which contains each
clause $C_\ell,E_\ell,r_{j_\ell}$
for $i_{\ell} \notin B^+$ such that $i_\ell \in S_{t_q}^B(C)$.
The subderivation~$L^\prime_\ell$ of~$L^\prime_{t_q}$
that ends with $C_\ell,E_\ell,r_{j_\ell}$ uses as resolution
variables precisely the
variables~$r_{j_{\ell^\prime}}$
such that $j_\ell\not=j_{\ell^\prime}$,
$i_{\ell^\prime}\in S_\ell^B(C)$, and $i_{\ell^\prime}\notin B^+$.

If $i_{t_q}\notin B_1$, then $t_q=1$ and $q=r$ and the final
clause of~$L^\prime_{t_q}$ is
$C_1,E_1,r_{j_1}$ (that is, with $\ell = t_q = 1$).
Suppose instead that $i_{t_q}\in B_1$ and $i_\ell$ is the (only)
predecessor of $i_{t_q}$ such that $i_\ell\in\Tdom(C^\Tpool)$.
If $i_\ell\in B^+$, then $S_{t_q}^B(C) \subseteq B^+$
and $L^\prime_{t_q}$ is empty.
If $i_\ell\notin B^+$,
then the final clause of $L^\prime_{t_q}$
is $C_\ell,E_\ell,r_{j_\ell}$.
\item[\rm (b)] There is a regRTI derivation~$L_{t_q}$ of
$C_{t_q},F_q,r_{j_{t_q}}$
of size $O(N_q^2 + m)$. The variables used as resolution
variables in~$L_{t_q}$ are at most the
variables~$r_j$
where $\negr_j\notin F_q\cup\{\negr_{t_q}\}$
and, if $i_{t_q}\in B_1$,
the variables $p_{i^\pprime,j}$ where
$i^\pprime$ is the predecessor of~$i_{t_q}$
such that $i^\pprime\notin \Tdom(C^\Tpool)$.
\end{description}
\end{lem}
Note that $i_{t_q}\in B^+\setminus B$ only if $t_q=1$,
and this is the case
where $L_1$ is not needed.
\begin{proof}
The regular dag-like refutation~$L^\prime_{t_q}$ for part~(a)
is constructed by induction on $\ell$ such that $i_\ell \in S_{t_q}^B(C) \setminus B^+$.
We add the clauses $C_\ell,E_\ell,r_{j_\ell}$ to~$L^\prime_{t_q}$ for larger
values of~$\ell$ first, making sure that the condition on resolution
variables remains satisfied. The inductive argument splits into four cases.

\emph{Case 1:} $i_\ell$ is bypassed by max-bypasser~$i_{\ell^\prime}$.
Since $i_{\ell}\notin B^+$, we have $i_{\ell^\prime}\notin B^+$.
The induction hypothesis tells us
that $C_{\ell^\prime},E_{\ell^\prime},r_{j_\ell^\prime}$ appears
in the already constructed portion of~$L^\prime_{t_q}$.
>From the remarks after the definitions
of $C_\ell$ and~$E_\ell$, we have $j_\ell = j_{\ell^\prime}$
and $C_\ell = C_{\ell^\prime}$ and $E_\ell = E_{\ell^\prime}$.
Thus $C_{\ell^\prime},E_{\ell^\prime},r_{j_\ell^\prime}$
is exactly the same clause as $C_\ell,E_\ell,r_{j_\ell}$.
So no further
resolution inferences need to be added to~$L^\prime_{t_q}$
to handle~$i_\ell$, and $L^\prime_\ell = L^\prime_{\ell^\prime}$.
Since $S^B_\ell(C) = \{i_\ell\} \cup S^B_{\ell^\prime}(C)$
and $i_{\ell}\notin B^+$,
the subderivation $L^\prime_\ell$ satisfies the condition about
which resolution variables are used.

The remaining cases all assume that
$i_\ell$ is not bypassed.

\emph{Case 2 (the base case):} both of $i_\ell$'s predecessors
$i_{\ell^\prime}$ and $i_{\ell^\pprime}$ are in~$B^+$.
Then $C_\ell$ is
$\negp_{i_{\ell^\prime},j_{\ell^\prime}},\negp_{i_{\ell^\pprime},j_{\ell^\pprime}},\negp_{i_\ell,j_\ell}$.
Also, $E_\ell$ is $\negr_{j_{\ell^\prime}},\negr_{j_{\ell^\pprime}}$.
Therefore the Stone clause~(\ref{eq:LprimeellStone}) is the
same as $C_\ell,E_\ell,r_{j_\ell}$, so $L^\prime_\ell$ consists
of just this clause.

\emph{Case 3:} neither of $i_\ell$'s predecessors
$i_{\ell^\prime}$ and $i_{\ell^\pprime}$ is in~$B^+$.
The already constructed part of~$L^\prime_\ell$
contains the clauses $C_{\ell^\prime},E_{\ell^\prime},r_{j_{\ell^\prime}}$
and $C_{\ell^\pprime},E_{\ell^\pprime},r_{j_{\ell^\pprime}}$.
Assume for the moment that $j_{\ell^\prime}\not= j_{\ell^\pprime}$.
W.l.o.g., the highest index $s^\pprime$ such that $j_{s^\pprime} = j_{\ell^\pprime}$
is strictly greater than
highest index $s^\prime$ such that $j_{s^\prime} = j_{\ell^\prime}$; otherwise
interchange $\ell^\prime$ and~$\ell^\pprime$.
It follows from the induction step
for Case~1 that~$L^\prime_{\ell^\pprime}$
equals ~$L^\prime_{s^\pprime}$.
Therefore, by the assumption that
$s^\prime < s^\pprime $  and the inductive
condition on resolution variables,
$r_{j_{\ell^\prime}}$
is not resolved on in~$L^\prime_{\ell^\pprime}$.

Using the Stone clause
\begin{equation}\label{eq:LprimeellStone}
\negp_{i_{\ell^\prime},j_{\ell^\prime}},\negr_{j_{\ell^\prime}},
\negp_{i_{\ell^\pprime},j_{\ell^\pprime}},\negr_{j_{\ell^\pprime}},
\negp_{i_\ell,j_\ell}, r_{j_\ell}
\end{equation}
form $L^\prime_\ell$ as
\begin{equation}\label{eq:Case2}
\AxiomC{$\negp_{i_{\ell^\prime},j_{\ell^\prime}},\negr_{j_{\ell^\prime}},
\negp_{i_{\ell^\pprime},j_{\ell^\pprime}},\negr_{j_{\ell^\pprime}},
\negp_{i_\ell,j_\ell}, r_{j_\ell}$}
\AxiomC{$\proofdots^{\textstyle L^\prime_{\ell^\pprime}}$}
\UnaryInfC{$C_{\ell^\pprime},E_{\ell^\pprime}, r_{j_{\ell^\pprime}}$}
\BinaryInfC{$\negp_{i_{\ell^\prime},j_{\ell^\prime}},\negr_{j_{\ell^\prime}},
\negp_{i_{\ell^\pprime},j_{\ell^\pprime}},C_{\ell^\pprime},E_{\ell^\pprime},
\negp_{i_\ell,j_\ell}, r_{j_\ell}$}
\AxiomC{$\proofdots^{\textstyle L^\prime_{\ell^\prime}}$}
\UnaryInfC{$C_{\ell^\prime},E_{\ell^\prime}, r_{j_{\ell^\prime}}$}
\BinaryInfC{$C_\ell,E_\ell, r_{j_\ell}$}
\DisplayProof
\end{equation}
We have $E_\ell = E_{\ell^\prime}\cup E_{\ell^\pprime}$.
The literals $\negp_{i_{\ell^\prime},j_{\ell^\prime}}$ and $\negp_{i_{\ell^\pprime},j_{\ell^\pprime}}$
may or may not appear in $C_{\ell^\prime}$ and
$C_{\ell^\pprime}$ (respectively), but in any case
$C_\ell$ is the same as
$\negp_{i_{\ell^\prime},j_{\ell^\prime}},C_{\ell^\prime},
\negp_{i_{\ell^\pprime},j_{\ell^\pprime}},C_{\ell^\pprime},
\negp_{i_\ell,j_\ell}$.
Note that the proof as pictured above might be slightly misleading:
we are constructing a dag-like derivation, not a tree-like
derivation and the
subderivations $L^\prime_{\ell^\prime}$
and $L^\prime_{\ell^\pprime}$ need not be
disjoint.  The regularity of $L^\prime_\ell$ follows from the induction
hypotheses and the fact that
$r_{j_{\ell^\prime}}$ is not a resolution variable of
$L^\prime_{\ell^\pprime}$.  The condition on which
resolution variables are used in~$L^\prime_\ell$ follows
from
$S^B_\ell(C) = \{i_\ell\}\cup S^B_{\ell^\prime}(C)\cup S^B_{\ell^\pprime}(C)$.

For the remaining part of case 3,
suppose that $j_{\ell^\prime}=j_{\ell^\pprime}$.
Let $s$ be the highest index such that $j_s = j_{\ell^\prime} = j_{\ell^\pprime}$.
By the remarks after the definitions of $C_\ell$ and~$E_\ell$,
we have
$C_{\ell^\prime} = C_{\ell^\pprime} = C_{s}$
and
$E_{\ell^\prime} = E_{\ell^\pprime} = E_{s}$.
Thus also
$L^\prime_{\ell^\prime} = L^\prime_{\ell^\pprime} = L^\prime_{s}$.
Now, argue as in the previous case, but replace
the inferences~(\ref{eq:Case2}) with the inference
\begin{prooftree}
\AxiomC{$\negp_{i_{\ell^\prime},j_{s}},
\negp_{i_{\ell^\pprime},j_{s}},\negr_{j_{s}},
\negp_{i_\ell,j_\ell}, r_{j_\ell}$}
\AxiomC{$\proofdots^{\textstyle L^\prime_{s}}$}
\UnaryInfC{$C_{s},E_{s}, r_{j_{s}}$}
\BinaryInfC{$C_\ell,E_\ell,r_{j_\ell}$}
\end{prooftree}
The left hypothesis is
a Stone clause.  The rest of the argument for this subcase is
as in the previous case.

\emph{Case 4:} $i_\ell$ has predecessors
$i_{\ell^\prime}\notin B^+$ and $i_{\ell^\pprime}\in B^+$.
The induction hypothesis says that
the clause $C_{\ell^\prime},E_{\ell^\prime},r_{j_{\ell^\prime}}$
has already been derived.  Then $C_\ell$ is
$\negp_{i_{\ell^\prime},j_{\ell^\prime}},
   \negp_{i_{\ell^\pprime},j_{\ell^\pprime}},
   \negp_{i_\ell,j_\ell},C_{\ell^\prime}$,
and $E_\ell$ is $\negr_{j_{\ell^\pprime}}, E_{\ell^\prime}$,
so we can form $L^\prime_\ell$ as
\begin{prooftree}
\AxiomC{$\negp_{i_{\ell^\prime},j_{\ell^\prime}},\negr_{j_{\ell^\prime}},
\negp_{i_{\ell^\pprime},j_{\ell^\pprime}},\negr_{j_{\ell^\pprime}},
\negp_{i_\ell,j_\ell}, r_{j_\ell}$}
\AxiomC{$\proofdots^{\textstyle L^\prime_{\ell^\prime}}$}
\UnaryInfC{$C_{\ell^\prime},E_{\ell^\prime}, r_{j_{\ell^\prime}}$}
\BinaryInfC{$C_\ell,E_\ell, r_{j_\ell}$}
\end{prooftree}
using resolution on $r_{\ell^\prime}$.
The regularity
of~$L^\prime_\ell$ and the conditions on which resolution variables
are used in~$L^\prime_\ell$ follow from the induction hypothesis
for~$L^\prime_{\ell^\prime}$ and the fact
that $S^B_\ell(C) = \{i_\ell,i_{\ell^\pprime}\}\cup S^B_{\ell^\prime}$.

That completes the proof of part~(a).  The size bound $O(N_q)$
on~$L^\prime_\ell$
follows from the fact that each of the four cases in the construction
of $L^\prime_\ell$ added $O(1)$ clauses.

To apply part~(a) in the proof of~(b), we need
regRTI derivations~$L^\pprime_{t_q}$,
instead of the regular dag-like
refutations~$L^\prime_{t_q}$.
For this,
Theorem~3.3 of~\cite{BHJ:ResTreeLearning} states that
the desired derivation~$L^\pprime_{t_q}$, containing exactly
the same clauses as~$L^\prime_{t_q}$, can be constructed
from~$L^\pprime_{t_q}$
with the size of~$L^\pprime_{t_q}$ bounded by the product of
the size and the height of
$L^\prime_{t_q}$.  Thus, the size of~$L^\pprime_{t_q}$
is $O(N_q^2)$.  Furthermore, $L^\pprime_{t_q}$~uses
the same resolution variables as~$L^\prime_{t_q}$.

We now prove part~(b).
First suppose that $q=r$ and $1\notin B$.
Then $t_q = 1$ and $F_q = E_1$.
In this case,
the regRTI derivation~$L^\pprime_1$ obtained from part~(a)
is already the desired derivation.
This only uses resolution variables~$r_{j_\ell}$
such that $i_{\ell}\notin B^+$ and hence
$j_\ell \notin \{j_{t_1},\ldots,j_{t_q}\}$.

Otherwise, $i_{t_q}\in B_1$,
and hence $F_q = \{\negr_{j_{t_1}},\ldots,\negr_{j_{t_{q-1}}}\}$.
Let $i_{\ell^\prime}$ and $i^\pprime$ be the
predecessors of~$i_{t_q}$; so,
$i^\pprime$ is 3-learned and $i^\pprime\notin \Tdom(C^\Tpool)$.
Note that $j_{\ell^\prime}$ may or may not be
in $\{j_{t_1},\ldots,j_{t_q}\}$.
Consider the $m-2$ many Stone clauses
for $j^\pprime\notin\{j_{t_q},j_{\ell^\prime}\}$:
\[
\negp_{i_{\ell^\prime},j_{\ell^\prime}},\negr_{j_{\ell^\prime}},
\negp_{i^\pprime,j^\pprime},\negr_{j^\pprime},
\negp_{i_{t_q},j_{t_q}}, r_{j_{t_q}}.
\]
Resolving these with the 3-learned clauses
$\negp_{i^\pprime,j^\pprime},r_{j^\pprime}$
for $j^\pprime\notin \{j_{t_1},\ldots,j_{t_{q}},j_{\ell^\prime}\}$
and the Stone clause $\bigvee_{j^\pprime} p_{i^\pprime,j^\pprime}$
gives the clause
\begin{equation}\label{eq:caseb}
\negp_{i_{\ell^\prime},j_{\ell^\prime}},\negr_{j_{\ell^\prime}},
F_q,\negp_{i_{t_q},j_{t_q}}, r_{j_{t_q}}
\end{equation}
by resolution on the variables $r_{j^\pprime}$
for $j^\pprime\notin \{j_{t_1},\ldots,j_{t_{q}},j_{\ell^\prime}\}$
and the variables $p_{i^\pprime,j^\pprime}$ for all~$j^\pprime$.

Suppose that $i_{\ell^\prime}\in B^+$ and thus $\negr_{j_{\ell^\prime}}\in F_q$
and $C_{t_q}$ is the
clause $\negp_{i_{\ell^\prime},j_{\ell^\prime}},\negp_{i_{t_q},j_{t_q}}$.
Then (\ref{eq:caseb}) is the same as
$C_{t_q},F_q,r_{j_{t_q}}$ and the construction of~$L_{t_q}$
for part~(b) is complete.  In this case, $L_{t_q}$ has size $O(m)$.

Alternately, suppose that $i_{\ell^\prime}\notin B^+$,
so $\negr_{j_{\ell^\prime}}\notin F_q$.
By the last assertion of part~(a), $L^\pprime_{t_q}$
is a regRTI derivation
of $C_{\ell^\prime},E_{\ell^\prime},r_{j_{\ell^\prime}}$
of size $O(N_q^2)$.
Form $L_{t_q}$ by resolving this
against~(\ref{eq:caseb})
on the variable~$r_{j_{\ell^\prime}}$ to obtain
$C_{t_q}, F_q, r_{j_{t_q}}$.
This is a valid resolution inference since
$C_{t_q}$ is
$\negp_{i_{t_q},j_{t_q}},\negp_{i_{\ell^\prime},j_{\ell^\prime}},C_{\ell^\prime}$,
and $E_{\ell^\prime}\subseteq F_q$.
The size of~$L_{t_q}$ is $O(N_q^2 + m)$.

This completes the description
of~$L_{t_q}$, and the proof of Lemma~\ref{lm:Ltq}.
\end{proof}

This also completes the construction
of the derivation~$R_C$ of Lemma~\ref{lm:closure}.
Since the construction of~$R_C$ invokes Lemmas
\ref{lm:LtqB2} and~\ref{lm:Ltq} at most $N$ times,
the size of~$R_C$ is bounded
by $O(N m^2 + N(N^2+m)) = O(N m^2)$.

\subsection*{Finishing the proof}

All that remains to finish the proof of Theorem~\ref{thm:regRTI}
is to bound the size of the refutation~$R$.
As described above, $R$~is built up from a derivation
containing only the empty clause
by applying the constructions of Figure~\ref{fig:extendible}
and Lemma~\ref{lm:closure}, always to the
leftmost currently unfinished leaf.

If the clause at that leaf is non-extendible,
it is dealt with using Lemma~\ref{lm:closure},
and no new unfinished
leaves appear.

Otherwise, the leftmost unfinished leaf contains an extendible clause,
and it is dealt with using the construction of Figure~\ref{fig:extendible},
leading to $m$ new unfinished leaves, at least one of them
%with a highest-index literal  $\negp_{i_{k+1},t}$ such that
labeled with a clause $C_t,\negp_{i_{k+1},t}$
such that $\negp_{i_{k+1},t}, r_t$ has
not been 3-learned.

A clause of the type allowed to appear in an unfinished leaf
can be iteratively extended at most $n$~times, so at some point
we have to reach a situation in which the construction of
Figure~\ref{fig:extendible}
produces only non-extendible clauses.
When Lemma~\ref{lm:closure} is
applied to those clauses, at least one $K$-learned clause of the form $\negp_{i,j}, r_j$
becomes $(K{+}1)$-learned.
There are $nm$ clauses of the form $\negp_{i,j}, r_j$ to be learned.
Once all of them have
been $3$-learned, all remaining unfinished leaves
become non-extendible and can be dealt with using
the construction of Lemma~\ref{lm:closure}.

Consider the set of clauses in~$R$ which at some
point were unfinished leaves during the construction
of~$R$.  Call such a clause {\em green}
if it was handled with Lemma~\ref{lm:closure}
and thereby a clause of the form $\negp_{i,j}, r_j$
became $(K{+}1)$-learned instead of
$K$-learned.
The remaining clauses are called {\em non-green}:
these clauses were either handled by Lemma~\ref{lm:closure}
but without any clause $\negp_{i,j}, r_j$ becoming $(K{+}1)$-learned, or were
handled with the construction of Figure~\ref{fig:extendible}.
These green and non-green clauses inherit a tree structure
from~$R$.
% Consider the tree in which all clauses in $R$ that
%at some point were unfinished leaves appear as nodes,
%with the clauses treated using the construction of Figure \ref{fig:extendible} as internal nodes,  the clauses
%treated using Lemma \ref{lm:closure} as leaves, and the tree structure given by %Figure \ref{fig:extendible}
%in the natural way. Call a leaf of this tree \emph{green} if the application of Lemma \ref{lm:closure} to it
%led to some clause becoming $(K{+}1)$-learned.
It follows from the discussion above that this tree
is $m$-branching, has depth at most~$n$, and
contains at most $3nm$ green leaves. Moreover,
each node in the tree is a green leaf,
is an ancestor or sibling of a green leaf,
or is the sibling of an ancestor
of a green leaf. It is straightforward to
prove that such a tree has at most $O(n^2m^2)$ leaves.
Since each of these leaves corresponds to
a single application of Lemma~\ref{lm:closure},
the size of~$R$ is at most $O(N^3m^4)$.
This completes the proof of Theorem~\ref{thm:regRTI}.
\end{proof}

Since the Stone tautologies
contain $O(Nm^3)$ many symbols, and since $N\le m$,
the size upper bound $O(N^3m^4)$ is quadratic
in the size of the clauses
being refuted.

The refutation~$R$ described above may not be greedy.
Although we lack a
proof, it is possible that
$R$~can be made greedy by omitting
subderivations of already learned clauses.
It is also possible that $R$ is, or could be made to be, unit-propagating.
In particular, note that
the only unit clauses that appear in~$R$ are the literals $\negp_{1,j}$
that appear as unfinished clauses in the very first step of the
construction of~$R$.  However, we have not tried to formally analyze the
greedy or unit-propagating properties of~$R$.

\section*{Acknowledgement}
\noindent
We thank the two anonymous referees for useful comments and suggestions.

\bigskip

\vbox{
\hfill{\it The proof is complete, If only I've stated it thrice.} \\
\hbox{\relax}\hfill {\sl Fit the Fifth -- The Beaver's Lesson, The Hunting of the Snark} \\
\hbox{\relax}\hfill Lewis Carroll \\
}

\bibliographystyle{alpha}
\bibliography{logic}

\end{document}